\documentclass[preprint,12pt]{elsarticle}

\usepackage{mathtools,dsfont}
\usepackage{psfrag,bbm,graphicx}
\usepackage{comment,balance}
\usepackage{epstopdf}
\usepackage{breakurl}
\usepackage{tikz}
\usepackage{url}
\usepackage[shortlabels]{enumitem}
\usepackage{epsfig,subfig}
\usepackage{multicol}
\usepackage{multirow}
\usepackage{amsfonts, color}
\usepackage{array}
\usepackage{amssymb}
\usepackage{amsthm}
\usepackage[linesnumbered,ruled,vlined]{algorithm2e}
\usepackage{nccmath}
\usepackage{bm}
\usepackage[hidelinks]{hyperref}
\allowdisplaybreaks

\newtheorem{corollary}{Corollary}
\newtheorem{lemma}{Lemma}
\newtheorem{remark}{Remark}

\newtheorem{theorem}{Theorem}

\newtheorem{assumption}{Assumption}

\newcommand{\remove}[1]{{}}
\newcommand\numberthis{\addtocounter{equation}{1}\tag{\theequation}}

\def\E{{\mathbb{E}}}
\def\Prob{{\mathbb{P}}}

\def \OO {\mathrm{O}}

\DeclareMathOperator*{\argmax}{arg\,max}
\DeclareMathOperator*{\argmin}{arg\,min}

\begin{document}

\begin{frontmatter}
\title{On the Regret of Online Edge Service Hosting \tnoteref{t1}}

\author{R Sri Prakash}
\ead{prakash.14191@gmail.com;sriprakash@ee.iitb.ac.in}

\author{Nikhil Karamchandani}
\ead{nikhilk@ee.iitb.ac.in}

\author{Sharayu Moharir} 
\ead{sharayum@ee.iitb.ac.in}

\affiliation{organization={Indian Institute of Technology Bombay},
	addressline={Powai},
	postcode={400076},
	city={Mumbai},
	country={India}}

\begin{abstract}
	We consider the problem of service hosting where a service provider can dynamically rent edge resources via short term contracts to ensure better quality of service to its customers. The service can also be partially hosted at the edge, in which case,  customers' requests can be partially served at the edge. 
	The total cost incurred by the system is modeled as a combination of the rent cost, the service cost incurred due to latency in serving customers, and the fetch cost incurred as a result of the bandwidth used to fetch the code/databases of the service from the cloud servers to host the service at the edge. In this paper, we compare multiple hosting policies with regret as a metric, defined as the difference in the cost incurred by the policy and the optimal policy over some time horizon $T$. In particular we consider the Retro Renting (RR) and Follow The Perturbed Leader (FTPL) policies proposed in the literature and provide performance guarantees on the regret of these policies. We show that under i.i.d stochastic arrivals, RR policy has linear regret while FTPL policy has constant regret. Next, we propose a variant of FTPL, namely Wait then FTPL (W-FTPL), which also has constant regret while demonstrating much better dependence on the fetch cost. We also show that under adversarial arrivals, RR policy has linear regret while both FTPL and  W-FTPL have regret \(\mathrm{O}(\sqrt{T})\) which is order-optimal.
\end{abstract}
\end{frontmatter}

\section{Introduction}
\label{sec:intro}

Software as a Service (SaaS) instances like online navigation platforms, Video-on-Demand services, etc., have stringent latency constraints in order to provide a good quality of experience to their customers. While most SaaSs use cloud resources, low latency necessitates the use of storage/computational resources at the edge, i.e., close to the end-user. A service is said to be hosted at the edge if the code and databases needed to serve user queries are stored on the edge servers and requests can be served at the edge. In this work, the service can also be partially hosted at the edge, in which case,  customers' requests can be partially served at the edge \cite{prakash2020partial, LN_partial_hosting}. For instance, consider a SaaS application which provides access to news articles on demand. A typical news article includes some text, some images and possibly some embedded videos. One way to partially host such a service is to store the text corresponding to the news articles at the edge and store the images/videos in the cloud. Each user request will then be served partially by the edge. 

We consider the setting where third-party resources can be rented via short-term contracts to host the service, and the edge hosting status of the SaaS can be dynamically changed over time. If the service is not hosted at the edge, it can be fetched from the cloud servers  by incurring a fetch cost. The performance of a hosting policy is a function of the rent cost, the fetch cost, and the quality of experience of the users.  We refer to the algorithmic challenge of determining what fraction of the service to host at the edge over time as the \emph{service hosting problem}. Note that we study the service hosting problem from the perspective of a specific SaaS provider which can rent third-party edge resources to improve their customers' quality of experience by reducing the latency of their service.  
\color{black}

Novel online service hosting policies with provable performance guarantees have been proposed in \cite{narayana2021renting, zhao2018red}.
The metric of interest in \cite{narayana2021renting,zhao2018red} is the \emph{competitive ratio}, defined as the ratio of the cost incurred by an online policy to the cost incurred by an optimal offline policy for the same request arrival sequence. Since the competitive ratio is multiplicative by definition, the absolute value of the difference in the cost incurred by an online policy and the optimal policy can be large even though the competitive ratio of the online policy is close to one. This motivates studying the performance of candidate policies in terms of \emph{regret}, defined as the difference in the cost incurred by the policy and the optimal policy. Regret is a widely used metric in online learning \cite{lai1985asymptotically}, including recently for the caching problem \cite{paschos2019learning,bhattacharjee2020fundamental}, which is closely related to the service hosting problem. In this work, one of our goals is to design hosting policies with provable guarantees in terms of the regret. Another key dimension in the performance guarantees of online policies is the assumption made on the request arrival process. Commonly studied settings include the \emph{stochastic setting} and the \emph{adversarial setting}, and
 we provide performance guarantees for both settings.


\subsection{Our Contributions}
We propose a variant of the FTPL policy called Wait then Follow the Perturbed Leader (W-FTPL). W-FTPL is a randomized policy that {takes into account the fetch cost in its decision-making}. More specifically, W-FTPL does not fetch the service for an initial wait period which depends on the request arrivals and is an increasing function of the fetch cost. Following the wait period, W-FTPL mimics the FTPL policy.

For i.i.d. stochastic request arrivals and the regret metric, we show that RR is sub-optimal while FTPL and W-FTPL are both order-optimal with respect to the time horizon. While the regret of FTPL can increase linearly with the fetch cost, the regret of W-FTPL increases at most polylogarithmically. The improved performance of W-FTPL over FTPL is a consequence of the fact that W-FTPL avoids most of the fetches made by FTPL in the initial time-slots and by the end of the wait period, its estimate of the arrival rate is accurate enough to avoid multiple fetches. 
Further, we characterize the competitive ratio of FTPL and W-FTPL for the setting where a finite number of partial hosting levels are allowed. We compare these results with the competitive ratio of RR when a single partial hosting level is allowed \cite{LN_partial_hosting} to conclude that  FTPL and W-FTPL outperform RR with respect to the competitive ratio.

For the adversarial setting, we first characterize a fundamental lower bound on the regret of any online hosting policy. In terms of regret, we then show that RR is strictly suboptimal, while FTPL and W-FTPL have order-optimal performance with respect to time. We show that the competitive ratios of FTPL and W-FTPL are upper bounded by a constant (with respect to time) for any finite number of partial hosting levels. This is a more general setting than the one considered in \cite{LN_partial_hosting}, where the competitive ratio of a variant of RR is characterized when only one partial hosting level is permitted. 

\subsection{Related work}
Several emerging applications such as Augmented / Virtual Reality (AR/VR), online gaming, the Internet of Things (IoT), and autonomous driving have very stringent latency requirements. In addition, many of these applications are also compute-and-storage heavy. This necessitates the migration of some of the storage and computation requirements of these services from remote cloud servers to the edge. To enable such offloading, several edge computing architectures have been proposed in the literature and their performance has been analyzed extensively, see for example \cite{Puliafito:2019, mao2017survey, mach2017mobile} and references therein. 
     
The service hosting problem has received a lot of attention recently, and various settings have been studied. One approach has been to pose the design problem of where to host which service as a one-shot large-scale optimization problem, see for example \cite{pasteris2019service,  bi2019joint, 9326369}. Our work differs from this line of work in that we consider an online setting where we adaptively decide when (and what fraction of) service to host at any time, depending on the varying number of requests seen thus far. 

There have been recent works \cite{zhao2018red, tan2021asymptotically} which consider the online service hosting problem, and design adaptive schemes whose performance is characterized in terms of the competitive ratio with respect to an oracle which known the entire request sequence apriori. Our work differs from these in a couple of ways. Firstly, the above mentioned works study the problem from the perspective of an edge resource provider who has to decide which amongst a collection of services to host at each time on a storage-constrained edge server. On the other hand, similar to some other works \cite{xia2020online, narayana2021renting}, we study the problem from the perspective of a service provider which needs to decide when to host its application at the edge so as to minimize the overall cost. Secondly, most of the works mentioned above do not consider partial hosting of the service, which is allowed in our framework. While \cite{LN_partial_hosting} does study partial hosting and proposed the Retro-Renting (RR) policy which is shown to have a constant competitive ratio, it only allows one level of partial hosting of the service. On the other hand, we consider a much more general setting where multiple partial service hosting levels are allowed and are able to show that both the FTPL and its variant W-FTPL are able to achieve a constant competitive ratio.  

While the works mentioned above study competitive ratio as a performance metric, another popular measure of performance for online algorithms is the regret which considers the difference in the cost incurred by an online policy and the optimal (static) policy. In particular, for the closely related caching problem, several policies inspired by the recent advances in online convex optimization have been proposed, including Online Gradient Ascent \cite{paschos2019learning}, Online Mirror Descent \cite{salem2021no}, and Follow the Perturbed Leader (FTPL) \cite{mukhopadhyay2021online,bhattacharjee2020fundamental}. In particular, FTPL has been shown to have order-optimal regret for the caching problem in the adversarial setting  \cite{mukhopadhyay2021online,bhattacharjee2020fundamental}.  In this work, we study regret for the RR and FTPL policies for the service hosting problem under both the stochastic and adversarial settings. 
Since RR is a deterministic policy, its regret performance in the adversarial setting is poor. While FTPL does achieve order-optimal regret (with respect to time), one limitation is that it makes hosting decisions {agnostic of the fetch cost}. As a result, in some cases, FTPL is prone to fetching and evicting the service multiple times in the initial time-slots when its estimate of the request arrival rate is noisy, thus leading to poor performance. This motivated our design and regret analysis of W-FTPL, a variant of FTPL which takes into account the fetch cost. One recent work which also considers regret as a performance metric is \cite{fan2022online}, which studied the problem of joint online service hosting and routing and proposed a two time-scale algorithm based on online gradient descent with provably order-optimal regret under adversarial requests. 

Finally, there are several other approaches which have been used to address the online service hosting problem, including Markov Decision Processes \cite{wang2015dynamic, xiong2101learning,prakash2020partial}, Reinforcement Learning \cite{zeng2019resource}, and Multi-Armed Bandits \cite{chen2019budget, he2020bandit}. 
\section{System Setup}
\label{sec:setting}

\color{black}
We consider a {system} consisting of a back-end server and an edge-server to serve customers of a service. The back-end server always hosts the service and can serve any requests that are routed to it.
In this work, we allow the service to be partially hosted at the edge-server.  When the service is partially hosted at the edge, requests are partially served at the edge and partially at the back-end server. 
Further, the fraction of service hosted at the edge can be changed over time. 
If the service is not hosted or partially hosted at the edge, parts of the service can be fetched from the back-end server to host at the edge.  

\emph{Sequence of events in each time-slot}: We consider a time-slotted system. In each time-slot, we first make the service hosting decision for that time-slot. Following this, requests may arrive and are served at the edge and/or the back-end server.

\emph{Request arrivals}: We consider two types of arrival processes in this paper. 
\begin{itemize}
	\item[-] \emph{Stochastic arrivals}: In this case, request arrivals are i.i.d.  across time-slots with mean $\mu$. 
	\item[-] \emph{Adversarial arrivals}: In this case, the request arrival sequence is generated by an oblivious adversary\footnote{ The entire request sequence is assumed to be fixed apriori.}.
\end{itemize}
Let $r_t$ denote the number of request arrivals in time-slot $t$. We consider bounded requests, specifically, $0\le r_t\le \kappa$, and let $R_t=\sum_{i=1}^{t}r_i$ denote the cumulative requests observed by the end of time-slot $t$. In both cases, we assume that there are at most $\kappa$ arrivals per slot.

\emph{Costs}: We model three types of costs.
\begin{enumerate}
	\item[-] \emph{Rent cost} ($\mathcal{C}^{\mathcal{P}}_{R,t}$): The system incurs a cost of $c$ units per time-slot to host the entire service at the edge. For hosting $f$ fraction of service, the system incurs a cost of $cf$ units per time-slot. 
	
	\item[-] \emph{Service cost} ($\mathcal{C}^{\mathcal{P}}_{S,t}$): This is the cost incurred to use the back-end server to serve (parts of) requests. 
	If $f$ fraction of service is hosted at the edge, the system incurs a service cost of $g(f)$ units per request, where $g(f)$ is a decreasing function of $f$. We fix $g(0)=1$, i.e., service cost is one unit when the service is not hosted at the edge.
	
	\item[-] \emph{Fetch cost} ($\mathcal{C}^{\mathcal{P}}_{F,t}$): The system incurs a cost of $M> 1$ units for each fetch of the entire service from the back-end server to host on the edge-server. For fetching $f$ fraction of service, the system incurs a cost of $Mf$ units.
\end{enumerate}

Let $\rho_t^{\mathcal{P}}$ denote the edge hosting status of the service in time slot $t$ under policy $\mathcal{P}$. We consider the setting where $\rho_{t}^\mathcal{P} \in \{\alpha_1,\alpha_2, \cdots , \alpha_{K-1},\alpha_{K}\}$ and $\rho_{t}^\mathcal{P}=\alpha_{i}$ implies that $\alpha_{i}$ fraction of the service is hosted at the edge-server in time-slot $t$. Note that $\alpha_{i}\in [0,1]$ and in particular we consider $\alpha_1=0$, $\alpha_K=1$ and $\alpha_{i}<\alpha_{j}$ for $i<j$. It follows that $\rho_t^\mathcal{P}=1$ implies that the entire service is hosted at the edge in time-slot $t$, and $\rho_t^\mathcal{P}=0$ implies that the service is not hosted at the edge in time-slot $t$.

We make some assumptions about the various system parameters. 
\begin{assumption}	\label{assump:k<c}
	$c\le \kappa$.
\end{assumption}
This assumption is motivated by the fact that if $\kappa < c$, it is strictly sub-optimal to host the entire service at the edge, irrespective of the number of arrivals in a time-slot.

\begin{assumption}\label{assump:g(f)_convex}
	For a fraction of service $\alpha_{i}$, $ 1\le i\le K$, $\alpha_{i}+g(\alpha_{i})\le 1 $.
\end{assumption}
 In \cite{prakash2020partial, LN_partial_hosting}, it is shown that the benefits of partial hosting are limited to the setting where $\alpha_{i}+g(\alpha_{i})\le 1 $. Motivated by this, we consider the case where  $\alpha_{i}+g(\alpha_{i})\le 1 $ for all candidate values of $\alpha_{i}$.

\begin{assumption} \label{assump:cost<k}
	For any fraction of service $\alpha_{i}$, $1\le i\le K$,
	$c\alpha_{i}+g(\alpha_{i})r_t\le \kappa$
\end{assumption}
 Assumption \ref{assump:cost<k} follows from  Assumptions \ref{assump:k<c} and \ref{assump:g(f)_convex}
 
\begin{assumption}
	For stochastic arrivals, $\mu>0$ and for adversarial arrivals, $R_T\ge 1$.
\end{assumption}
This assumption eliminates degenerate cases where there are no request arrivals in the time-horizon of interest. 

The total cost incurred in time-slot $t$ by policy $\mathcal{P}$, denoted by $\mathcal{C}^{\mathcal{P}}_t(r_t)$, is the sum of the rent, service, and fetch costs. It follows that,
\begin{align*}
\mathcal{C}^{\mathcal{P}}_t(r_t)
&=c\rho_{t}^{\mathcal{P}}+ g(\rho_{t}^{\mathcal{P}}) r_t+ M (\rho_{t}^{\mathcal{P}}-\rho_{t-1}^{\mathcal{P}})^+. \numberthis \label{eq:cost_slot}
\end{align*}
Let $r=\{r_t\}_{t\ge1}$ denote the request arrival sequence and {$\mathcal{C}^{\mathcal{P}}(T,r)$} denote the cumulative cost incurred by policy $\mathcal{P}$ in time-slots 1 to $T$. It follows that, 
$$
\mathcal{C}^{\mathcal{P}}(T,r) = \sum_{t=1}^T \mathcal{C}^{\mathcal{P}}_t(r_t).
$$

\emph{Performance metrics}: We consider two performance metrics, namely regret and competitive ratio.

For i.i.d. stochastic arrivals, the regret of a policy $\mathcal{P}$, denoted by $\mathcal{R}^{\mathcal{P}}_S(T)$, 
is defined as the difference in the total expected cost incurred by the policy and the oracle static hosting policy. The oracle static hosting policy makes a hosting decision at $t=1$ using the knowledge of the statistics of the request arrival process but not the entire sample-path. For ease of notation, we define $\mu=\E[r_t]$, $\mu_i=c\alpha_i+g(\alpha_i)\mu$. It follows that the expected cost incurred by the optimal static hosting policy in time-slots 1 to $T$ is $\min_i\{c\alpha_iT+g(\alpha_i)\mu T+M\alpha_i\}$, and therefore, 
\begin{align*}
\mathcal{R}^{\mathcal{P}}_S(T)&=  \E_{\mathcal{P},r}[\mathcal{C}^\mathcal{P}(T,r)] -\min_i\{c\alpha_iT+g(\alpha_i)\mu T+M\alpha_i\}\\
&=\E_{\mathcal{P},r}[\mathcal{C}^\mathcal{P}(T,r)] -\min_i\{\mu_i T+M\alpha_i\},\\
& \le \E_{\mathcal{P},r}[\mathcal{C}^\mathcal{P}(T,r)] -\min_i\{\mu_i T\}, \numberthis \label{eq:Regret_stch} 
\end{align*} 

where $\E_{\mathcal{P},r}$ represents the expectation over the randomization in policy $\mathcal{P}$ and the request arrival sequences $r$. 

For i.i.d. stochastic arrivals, the competitive ratio of a policy $\mathcal{P}$, denoted by $\sigma^{\mathcal{P}}_S(T)$, 
is defined as the ratio of the expected total cost incurred by the policy and expected total cost incurred by the optimal static hosting policy as discussed above. It follows that,
\begin{align*}
\sigma_{S}^{\mathcal{P}}(T)= \frac{\E_{\mathcal{P},r}[\mathcal{C}^{\mathcal{P}}(T,r)]}{\min_i\{c\alpha_iT+g(\alpha_i)\mu T+M\alpha_i\}}.
\end{align*}

For adversarial arrivals, the regret of a policy $\mathcal{P}$, denoted by $\mathcal{R}^{\mathcal{P}}_A(T)$, is defined as the worst case difference in the total expected cost incurred by the policy and the optimal static hosting policy. The expectation is taken over the randomization in policy $\mathcal{P}$. It follows that for any request sequence $r$, the total cost incurred by the optimal static hosting policy in time-slots 1 to $T$ is $\min_i\{c\alpha_iT+g(\alpha_i)R_T+M\alpha_i\}$, and therefore, 
\begin{align*}
\mathcal{R}^{\mathcal{P}}_A(T,r)&=\E_\mathcal{P}[\mathcal{C}^\mathcal{P}(T,r)] - \min_i\{c\alpha_iT+g(\alpha_i)R_T+M\alpha_i\},\\
\mathcal{R}^{\mathcal{P}}_A(T)&=\sup_{r} \left(\mathcal{R}^{\mathcal{P}}_A(T,r) \right). 
\end{align*}

We have the following relation between the regret induced by a policy under i.i.d stochastic arrivals and adversarial arrivals.
\begin{lemma}\label{lem:adv_stch_reg}
	For any online policy $\mathcal{P}$ and stochastic arrivals we have,
	$\mathcal{R}_S^\mathcal{P}(T)\le \mathcal{R}_A^\mathcal{P}(T)$.
\end{lemma}
The proof can be found in Section \ref{sec:appendix}.

\color{black}
For adversarial arrivals, the competitive ratio of a policy $\mathcal{P}$, denoted by $\sigma^{\mathcal{P}}_A(T)$, is defined as the worst case ratio of the expected total cost incurred by the policy and the cost incurred by the offline optimal policy. The optimal offline policy can change the hosting status over time using the knowledge of the entire request process a priori. Let the cost incurred by the optimal offline hosting policy in time-slots 1 to $T$ be denoted by $\mathcal{C}^{\text{OFF-OPT}}(T,r)$. It follows that, 
\begin{align*}
\sigma^{\mathcal{P}}_A(T)=\sup_{r} \frac{\E_{\mathcal{P}}[\mathcal{C}^{\mathcal{P}}(T,r)]}{\mathcal{C}^{\text{OFF-OPT}}(T,r)}.
\end{align*}

\emph{Goal}: The goal is to design online hosting policies with provable performance guarantees with respect to regret and competitive ratio.

\emph{Notations}: We define some notations which will be used in the rest of this paper. Let $\bm{s}=[0 , \alpha_2, \alpha_3, \cdots, 1]'$, $\bm{f}=[1,g(\alpha_2), g(\alpha_3), \cdots, 0]'$, and $\bm{\rho}_t^{\mathcal{P}}\in \mathcal{X}$, where $\mathcal{X}$ is the collection of all possible one hot vectors in $\{0,1\}^K$  and thus $|\mathcal{X}|=K$. The position of 1 in the one hot vector $\bm{\rho}_t^{\mathcal{P}}$ represents the level of service hosted at the edge in slot $t$, that is $\rho_t^{\mathcal{P}}=\langle \bm{\rho}_t^{\mathcal{P}},\mathbf{s} \rangle$. 
The total cost in slot $t$ as given in \eqref{eq:cost_slot} can be rewritten  using the above notations as follows,
\begin{align*}
\mathcal{C}_{t}^{\mathcal{P}}(r_t) &= c\langle \bm{\rho}_t^{\mathcal{P}},\bm{s} \rangle + r_t \langle \bm{\rho}_t^{\mathcal{P}},\bm{f} \rangle + M (\rho_t^{\mathcal{P}}-\rho_{t-1}^{\mathcal{P}})^+,\\
&=\langle \bm{\rho}_t^{\mathcal{P}},\bm{\theta}_{t} \rangle + M (\rho_t^{\mathcal{P}}-\rho_{t-1}^{\mathcal{P}})^+,
\end{align*}
where $\bm{\theta}_t= c\bm{s}+r_t\bm{f}$. 

For ease of notation, for stochastic arrivals we define $\mu = \E[r_t]$, $\mu_i=c\alpha_{i}+g(\alpha_{i})\mu$, $\Delta_{ij}=\mu_i-\mu_j$, $i^\star=\argmin_i \mu_i$ , $\Delta_i=\mu_i-\mu_{i^\star}$, $\Delta_{\text{min}}=\min_i \Delta_{i}$, $\Delta_{\text{max}}=\max_i \Delta_{i}$.

We summarize important notations used in this paper in \tablename~\ref{tab:notations}.
\color{black}

\begin{table}
\begin{tabular}{|c|l|}
\hline
Symbol & Description \\
\hline \hline
$c$ & Rent cost for hosting complete service in a slot \\
\hline 
$M$ &	Fetch cost \\		
\hline
$K$ & Number of hosting levels allowed\\
\hline
$t$ & Time index\\
\hline
$\rho_{t}^\mathcal{P}$ & Fraction of service hosted in time slot $t$ under policy $\mathcal{P}$\\
\hline
$\kappa$ & Maximum number of arrivals in a slot\\
\hline
$\alpha_i$ & Fraction of service \\
\hline
$g(f)$ & Fraction of request forwarded if $f$ fraction is hosted \\
\hline
$r_t$ & Number of requests in time slot $t$\\
\hline
$R_t$ & Cumulative number of requests up to time slot $t$\\
\hline
$\mu$ & $\E[r_t]$\\
\hline
$T$ & Time horizon\\
\hline
$\mathcal{C}^\mathcal{P}(T,r)$ & Cost incurred till $T$ under policy $\mathcal{P}$ for request sequence $r$  \\
\hline
$\mathcal{R}_S^\mathcal{P}(T)$ & Regret under stochastic arrival setting\\
\hline
$\sigma_{S}^\mathcal{P}(T)$ & Competitive ratio under stochastic arrival setting\\ 
\hline
$\mathcal{R}^\mathcal{P}_A(T)$ & Regret under adversarial arrival setting\\ 
\hline
$\sigma_{A}^\mathcal{P}(T)$ & Competitive ratio under adversarial arrival setting\\ 
\hline
$\mu_i$ & $c\alpha_{i}+g(\alpha_{i})\mu$\\
\hline
$\Delta_{ij}$ & $\mu_i-\mu_j$\\
\hline
$i^\star$ & $\argmin_i \mu_i$ \\ 
\hline
$\Delta_{i}$ & $\mu_i-\mu_{i^\star}$ \\ 
\hline
$\Delta_{\text{min}}$ &  $\min_i \Delta_i$\\ 
\hline
$\Delta_{\text{max}}$ &  $\max_i \Delta_i$\\ 
\hline
$\bm{s}$ & $[0 , \alpha_2, \alpha_3, \cdots, 1]'$\\ 
\hline
$\bm{f}$ & $[1,g(\alpha_2), g(\alpha_3), \cdots, 0]'$\\ 
\hline
$\bm{\theta}_{t}$ &  $c\bm{s}+r_t\bm{f}$\\ 
\hline
$\bm{\Theta}_{t}$ &  $\sum_{i=1}^{t-1}\bm{\theta}_i$\\ 
\hline
$\eta_{t}$ & Learning rate\\ 
\hline
$\alpha$ & Hyper parameter in FTPL, W-FTPL policies\\
\hline
$\beta$, $\delta$ & Hyperparameters in W-FTPL policy\\ 
\hline
$\bm{\gamma}$ & IID Gaussian vector\\ 
\hline
$\mathcal{X}$ & Collection of one hot vectors in $\{0,1\}^K$\\ 
\hline
$T_s$ & Waiting time for W-FTPL policy\\
\hline
\end{tabular}
\caption{Notations}
\label{tab:notations}
\end{table}

\section{Policies}
\label{sec:policies}

\emph{Our Policy}: Our policy called Wait then Follow the Perturbed Leader (W-FTPL), is a variant of the popular FTPL policy. A formal definition of FTPL is given in Algorithm \ref{alg:FTPL}. FTPL is a randomized policy and is known to perform well for the traditional caching problem, in fact achieving order-wise optimal regret for adversarial arrivals \cite{bhattacharjee2020fundamental}.  Under FTPL, in any time slot $t$, for each $i$, FTPL considers the cumulative cost that would have been incurred had the system used $\alpha_i$ hosting fraction in the entire duration $[1, t -1]$, and then perturbs this by an  appropriately scaled version of a sample of a standard Gaussian random variable. FTPL then hosts the fraction of service for which the perturbed cost is minimum in that slot.


 
\begin{algorithm} 
	\caption{Follow The Perturbed Leader (FTPL)}
	\label{alg:FTPL}
	\KwIn{$c$, $\{\eta_t\}_{t\ge 1}$, $\{r_t\}_{t\ge 1}$}
	Set $\bm{\Theta}_{1} \leftarrow \bm{0}$\\
	Sample $\bm{\gamma} \sim \mathcal{N}(\bm{0}, \bm{I})$\\
	\For{$t=1$ to $T$}{
		host $\bm{\rho}_t\in \argmin_{\bm{\rho} \in \mathcal{X}} \langle \bm{\rho}, \bm{\Theta}_t +\eta_t\bm{\gamma} \rangle $\\ 
		$\bm{\theta}_{t}=c\bm{s}+r_t\bm{f}$\\
		Update $\bm{\Theta}_{t+1}=\bm{\Theta}_t+ \bm{\theta}_{t}$}
\end{algorithm}

We propose a policy called Wait then Follow The Perturbed Leader (W-FTPL). The key idea behind the W-FTPL policy is to not host the service for an initial wait period. This is to reduce the number of fetches made initially when the estimate of the arrivals is noisy.  The duration of this wait period is not fixed apriori and is a  function of the arrival pattern seen till that time. Following the wait period, W-FTPL mimics the FTPL policy. Refer to Algorithm \ref{alg:W-FTPL} for a formal definition of W-FTPL. Let $T_s$ be the time slot, after which W-FTPL starts acting as FTPL. Formally we define $T_s=\min\{t:t<\frac{(\max_{i\ne j}(\Theta_{t+1,i}-\Theta_{t+1,j}))^2}{\kappa^2\beta (\log M)^{1+\delta}}\}$, where $\beta>1$, $\delta>0$ are constants.

\begin{algorithm} 
	\caption{Wait then Follow The Perturbed Leader (W-FTPL)}
	\label{alg:W-FTPL}
	\KwIn{$c$, $M$, $\{\eta_t\}_{t\ge 1}$, $\{r_t\}_{t\ge1}$, $\beta$, $\delta$, $\kappa$}
	Set $\bm{\Theta_{1}} \leftarrow \bm{0}$, wait=1\\
	Sample $\bm{\gamma} \sim \mathcal{N}(\bm{0}, \bm{I})$\\
	\For{$t=1$ to $T$}
{  	\eIf{$\text{wait}=1$}{$\rho_t=0$}{
	host $ \bm{\rho}_t\in \argmin_{\bm{\rho} \in \mathcal{X}} \langle \bm{\rho}, \bm{\Theta}_t +\eta_t\bm{\gamma} \rangle $} 
	$\bm{\theta}_{t}=c\bm{s}+r_t\bm{f}$\\
	Update $\bm{\Theta}_{t+1} = \bm{\Theta}_{t} + \bm{\theta}_t$\\
	$\text{wait} = \min\{\text{wait},\mathbbm{1}_{t\ge \frac{(\max_{i\ne j}(\Theta_{t+1,i}-\Theta_{t+1,j}))^2}{\kappa^2\beta (\log M)^{1+\delta}}}\}$ 
}
\end{algorithm}

\emph{Retro-Renting (RR)} \cite{narayana2021renting}: The RR policy is a deterministic hosting policy, and it either hosts the complete service or hosts nothing in a slot. The key idea behind this policy is to use recent arrival patterns to make hosting decisions. We refer the reader to Algorithm 1 in \cite{narayana2021renting} for a formal definition of the RR policy. The performance of RR with respect to the competitive ratio was analyzed in \cite{narayana2021renting}. 

$\alpha-$RR \cite{LN_partial_hosting}  is a variant of RR where one partial level of hosting  the service is allowed. The formal definition of $\alpha-$RR, borrowed from \cite[Algorithm 1]{LN_partial_hosting} is given in Algorithm \ref{alg:alpha_RR}. Similar to RR, the key idea behind $\alpha-$RR policy is to check whether the current hosting status is optimal in the hindsight using recent arrival patterns to make hosting decisions. $\alpha-$RR checks the optimal offline policy cost for the recent request arrivals and chooses the fraction of service with the minimum cost. For $\rho_{t}^\mathcal{P}\in\{0,1\}$, $\alpha-$RR behaves as RR \cite{LN_partial_hosting}. Therefore we refer $\alpha-$RR as RR in this paper.

\begin{algorithm}
	\caption{$\alpha$-RetroRenting ($\alpha$-RR)}\label{alg:alpha_RR}
	\SetAlgoLined
	\SetKwFunction{FtotalCost}{totalCost}
	\SetKwProg{Fn}{Function}{:}{end}
	Input: Fetch cost $M$, partial hosting level $\alpha_2$, latency cost under partial hosting $g(\alpha_2)$,  rent cost $c$, request arrival sequence $\{x_l\}_{l\geq0}^t$\\
	Output:  service hosting strategy $r_{t+1}$, $t > 0$\\
	Initialize:  $r_1= t_{\text{recent}} = 0$\\
	\For {\textbf{each} time-slot $t$}{
		$I_t=(M$, $g(\alpha_2)$, $t$,  $t_{\text{recent}}$, $c$, $\{x_l\}_{l\geq t_{\text{recent}}}^t )$\\
		$R_0^{(\tau_0)} = [\underbrace{r_t,r_t,\ldots,r_t}_{\tau_0-t_{\text{recent}}},\underbrace{0,0,\ldots, 0}_{t-\tau_0}]$ \\
		$R_\alpha^{(\tau_{\alpha})}= [\underbrace{r_t,r_t,\ldots,r_t}_{\tau_{\alpha}-t_{\text{recent}}},\underbrace{\alpha_2,\alpha_2,\ldots, \alpha_2}_{t-\tau_{\alpha}}]$ \\
		$R_1^{(\tau_1)} = [\underbrace{r_t,r_t,\ldots,r_t}_{\tau_1-t_{\text{recent}}},\underbrace{1,1,\ldots, 1}_{t-\tau_1}]$ \\
		$\text{minCost}(0)=\displaystyle\min_{\tau_0\in (t_{\text{recent}},t)} \FtotalCost(R_0^{(\tau_0)},I_t)$\\
		$\text{minCost}(\alpha_2)=\displaystyle\min_{\tau_{\alpha}\in (t_{\text{recent}},t)} \FtotalCost(R_\alpha^{(\tau_{\alpha})},I_t)$\\
		$\text{minCost}(1)=\displaystyle\min_{\tau_1\in (t_{\text{recent}},t)} \FtotalCost(R_1^{(\tau_1)},I_t)$\\
		$r_{t+1} = \displaystyle \argmin_{i \in \{0, \alpha_2, 1\}} \text{minCost}(i)$ \\
		\If{$r_{t+1} \neq r_t$}
		{$t_{\text{recent}} =  t$
		}
	}
	\Fn{\FtotalCost{$R, I_t$}}{
		$g(0)=1$, $g(1)=0$\;
		cost = $R(1)\times c+x_1 \times g(R(1))$\;
		\For{$j\gets 2$ \KwTo $t-t_{\text{recent}}$ }{
			cost = cost $+ R(j)\times c+x_j \times g(R(j))$ \\
			\hspace{.67in}$+ M \times \left|R(j)-R(j-1)\right|$\;
		}
		\KwRet cost;
	}
\end{algorithm}

\section{Main Results and Discussion}
\label{sec:mainResults}
In this section, we state and discuss our key results. 
\subsection{Regret Analysis}
Our first result characterizes the regret performance of the policies discussed in Section \ref{sec:policies} and the fundamental limit on the performance of any online policy for adversarial arrivals. 
\begin{theorem}[Adversarial Arrivals]	\label{thm:reg_adv}
	Let the arrivals be generated by an oblivious adversary under the constraint that at most $\kappa$ requests arrive in each time-slot. Then,
	\begin{enumerate}[(a)]
		\item $
		\mathcal{R}^{\mathcal{P}}_{A}(T) =\Omega(\sqrt{T})$ for any online policy $\mathcal{P}$.
		\item $\mathcal{R}^{\text{RR}}_A(T) =\Omega({T}).$
		\item For $\eta_{t}=\alpha\sqrt{t}$, $\alpha>0$, 
		\begin{align*}
			\mathcal{R}^{\text{FTPL}}_A(T)\le& \sqrt{2T\log K}\left( \alpha  + \frac{4\kappa^2}{\alpha}\right)
			 +\frac{K^2M(c+2\kappa)}{2\alpha\sqrt{\pi}} \sqrt{T+1}.
		\end{align*}
		\item For $\eta_{t}=\alpha\sqrt{t}$, $\alpha>0$, 
		$$\mathcal{R}_A^{\text{FTPL}}(T) \ge  M\alpha_{2}/K.$$ \color{black}
		\item For $\eta_{t}=\alpha\sqrt{t}$, $\alpha>0$, $\beta>1$, $\delta>0$,  
		 \begin{align*}
				\mathcal{R}^{\text{W-FTPL}}_A(T)\le&  \sqrt{\beta T (\log M)^{1+\delta}}
				+\sqrt{2T\log K} \left(\alpha + \frac{4\kappa^2}{\alpha} \right)\\
				&+\frac{K^2M(c+2\kappa)}{2\alpha\sqrt{\pi}}\sqrt{T+1}.
			\end{align*}	
\end{enumerate}
\end{theorem}

The key take-away from Theorem \ref{thm:reg_adv} is that RR has linear regret and FTPL, W-FTPL have order-optimal regret, i.e., $\OO(\sqrt{T})$ with respect to the time horizon under adversarial arrivals. The proof of Theorem \ref{thm:reg_adv} is in Section \ref{sec:proofs}.

Our second result characterizes the regret performance of the policies discussed in Section \ref{sec:policies} for i.i.d. stochastic arrivals. 
\begin{theorem}[Stochastic Arrivals]\label{thm:reg_stch}
	Let the arrivals in each time-slot be i.i.d. stochastic with mean $\mu$ and $\Delta_{\text{min}}\ne 0$. Then we have
	\begin{enumerate}[(a)]
		\item $\mathcal{R}^{\text{RR}}_S(T) =\Omega({T}).$
		\item For $\eta_{t}=\alpha\sqrt{t-1}$, $\alpha>0$,
		\begin{align*}
		\mathcal{R}^{\text{FTPL}}_S(T) 
		\le& \left( \sqrt{2\log K}+\frac{2\sqrt{2h_1}\log K}{\Delta_{\text{min}}} \right) \left( \alpha  + \frac{4\kappa^2}{\alpha}\right)
		+ \frac{16\alpha^2+4\kappa^2}{\Delta_{\text{min}}}\\ 
		&+(16\alpha^2+3\kappa^2) \frac{M}{\Delta_{\text{min}}^2},
		\end{align*} 
		where $h_1= 4\max\{8\alpha^2,\kappa^2\}$.
	\item For $\eta_{t}=\alpha\sqrt{t-1}$, $\alpha>0$, $\beta>1$, $\delta>0$,  and $M$ large enough we have,
		\begin{align*}
			\mathcal{R}^{\text{W-FTPL}}_S(T)
			\le  1&+\beta \kappa^2(\log M)^{1+\delta}\left(\frac{4}{\Delta_{\text{min}}^2}+\frac{16\alpha^2+3\kappa^2}{\Delta_{\text{min}}^2 \Delta_{\text{max}}^2}\right)\\
			&+(16\alpha^2+4\kappa^2)\left(\frac{1}{\Delta_{\text{min}}}+\frac{1}{\Delta_{\text{min}}^2}\right)\\
			&+\left( \sqrt{2\log K}+\frac{2\sqrt{2h_1}\log K}{\Delta_{\text{min}}} \right) \left( \alpha  + \frac{4\kappa^2}{\alpha}\right),
		\end{align*}
		where $h_1= 4\max\{8\alpha^2,\kappa^2\}$.
	\end{enumerate}
\end{theorem}

\begin{corollary}[Stochastic Arrivals]
	Let the arrivals in each time-slot be i.i.d. stochastic with mean $\mu$ and $\Delta_{\text{min}}\ne 0$. For the W-FTPL policy, if $\rho_{t}^\mathcal{P}\in\{0,1\}$,  $\beta=\max\{(1+4\alpha/\kappa)^2, (1+\sqrt{2}/\kappa)^2\}$, then the regret scales at most logarithmically with $M$.
\end{corollary}

\begin{remark} While the result in \ref{thm:reg_stch} $(c)$ is stated for $\delta>0$, it can e shown that for $\delta=0$ and $\beta$ chosen suitably large, the  regret of W-FTPL scales logarithmically with $M$.
\end{remark}
\color{black}
%
%

The key take-aways from Theorem \ref{thm:reg_stch} are that RR has linear regret with respect to time, and FTPL/W-FTPL has regret that does not scale with time. Also, the upper bound on the regret of W-FTPL scales as $(\log M)^{1+\delta}$ for $\delta\ge 0$ where as for FTPL, the upper bound on the regret scales linearly with $M$. W-FTPL thus outperforms FTPL w.r.t. to the dependence of regret on the fetch cost $M$.  We validate these results via simulations in Section \ref{sec:simulations}.

\color{black}

\subsection{Competitive Ratio Analysis}
Our next result characterizes the competitive ratio of the policies discussed in Section \ref{sec:policies}. The competitive ratio of RR is characterized in \cite{narayana2021renting, LN_partial_hosting}, and it is shown that RR has a constant competitive ratio and the competitive ratio decreases with an increase in fetch cost $M$. We present the results for FTPL and W-FTPL policies below.
\begin{theorem}[Adversarial Arrivals]\label{thm:comp_adv}
	Let the arrivals be generated by an oblivious adversary under the constraint that at most $\kappa$ requests arrive in each time-slot. Then we have
\begin{enumerate}[(a)]
	\item For $\eta_{t}=\alpha\sqrt{t-1}$, $\alpha>0$, 
	\begin{align*}
	\sigma^{\text{FTPL}}_{A}(T)\le& \frac{\kappa^2(3+2M/c)}{\min_i(c\alpha_{i}+g(\alpha_{i}) \kappa)}\max_{\alpha_{i}\ne 0}\frac{1-g(\alpha_{i})}{\alpha_{i}} \\&+\frac{\kappa^2(M+c)}{\min_i(c\alpha_{i}+g(\alpha_{i}) \kappa)}\sum_{i=2}^K\frac{16\alpha^2}{c^2\alpha_{i}^2}.
	\end{align*}
	\item For $\eta_{t}=\alpha\sqrt{t-1}$, $\alpha>0$, $\beta>1$, $\delta\ge0$,
	\begin{align*}
	\sigma^{\text{W-FTPL}}_{A}(T)\le& \frac{\kappa^2(3+2M/c)}{\min_i(c\alpha_{i}+g(\alpha_{i}) \kappa)}\max_{\alpha_{i}\ne 0}\frac{1-g(\alpha_{i})}{\alpha_{i}}\\
	&+\frac{\kappa^2(M+c)}{\min_i(c\alpha_{i}+g(\alpha_{i}) \kappa)}\sum_{i=2}^K\frac{16\alpha^2}{c^2\alpha_{i}^2}\\
	&+\frac{\kappa^2}{\min_i(c\alpha_{i}+g(\alpha_{i}) \kappa)}.
	\end{align*}
	\end{enumerate}
\end{theorem}

The key take-away from Theorem \ref{thm:comp_adv} is that the competitive ratios of both FTPL and W-FTPL are bounded by a constant. The competitive ratio of FTPL and W-FTPL deteriorates with $M$, while the competitive ratio of RR improves with $M$ when only one partial hosting level is permitted \cite{LN_partial_hosting}. Note that the competitive ratio of FTPL and W-FTPL is bounded by a constant for the setting where any finite number of partial hosting levels are allowed. Compared to this, for RR in \cite{LN_partial_hosting}, the competitive ratio result holds only for one partial hosting level.

\begin{theorem}[Stochastic Arrivals] \label{thm:comp_stch}
	Let the arrivals in each time-slot be i.i.d. stochastic with mean $\mu$. Then we have
\begin{enumerate}[(a)]
	\item For $T$ large enough, $\sigma^{RR}_S (T)>  1$.
	\item For $\eta_{t}=\alpha \sqrt{t-1}$, $\sigma^{\text{FTPL}}_S(T) \le 1+\OO(1/T) $.
	\item For $\eta_{t}=\alpha\sqrt{t-1}$, $\alpha>0$, $\beta>1$, $\delta>0$, $\sigma^{\text{W-FTPL}}_S(T) \le 1+\OO(1/T) $.
\end{enumerate}	
\end{theorem}
The key take-away from Theorem \ref{thm:comp_stch} is that for $T$ large enough, FTPL and W-FTPL outperform RR in terms of competitive ratio for stochastic arrivals.
 
%
%
  In Section \ref{sec:simulations} we supplement the results stated in this section through simulations.
\section{Simulations}
\label{sec:simulations}
In this section, we compare RR, FTPL and W-FTPL policies via simulations using synthetic request sequences, as well as real trace data from \cite{iosup2008grid}. In all our simulations, we consider $\eta_t=0.1\sqrt{t}$ for FTPL and W-FTPL, and $\beta=6$ for W-FTPL.

\subsection{Synthetic arrivals}
In \figurename ~\ref{fig:RR_worst}, we plot the regret of different policies as a function of time horizon $T$. The parameters considered are $c=0.1$, $M=50$, $\kappa=5$. The request arrival considered is similar to the one used in the proof of Theorem \ref{thm:comp_adv}(b), i.e., we divide time into frames, and each frame starts with $\lceil \frac{M}{\kappa-c}\rceil$ slots of $\kappa$ requests each followed by  $\lceil \frac{M}{c}\rceil$ slots of zero requests. We consider 100 such frames for this experiment. The regret is averaged over 50 experiments. We observe that RR has linear regret in this case which agrees with Theorem \ref{thm:reg_adv}.

\begin{figure}[ht]
	\centering
	\includegraphics[width=\linewidth]{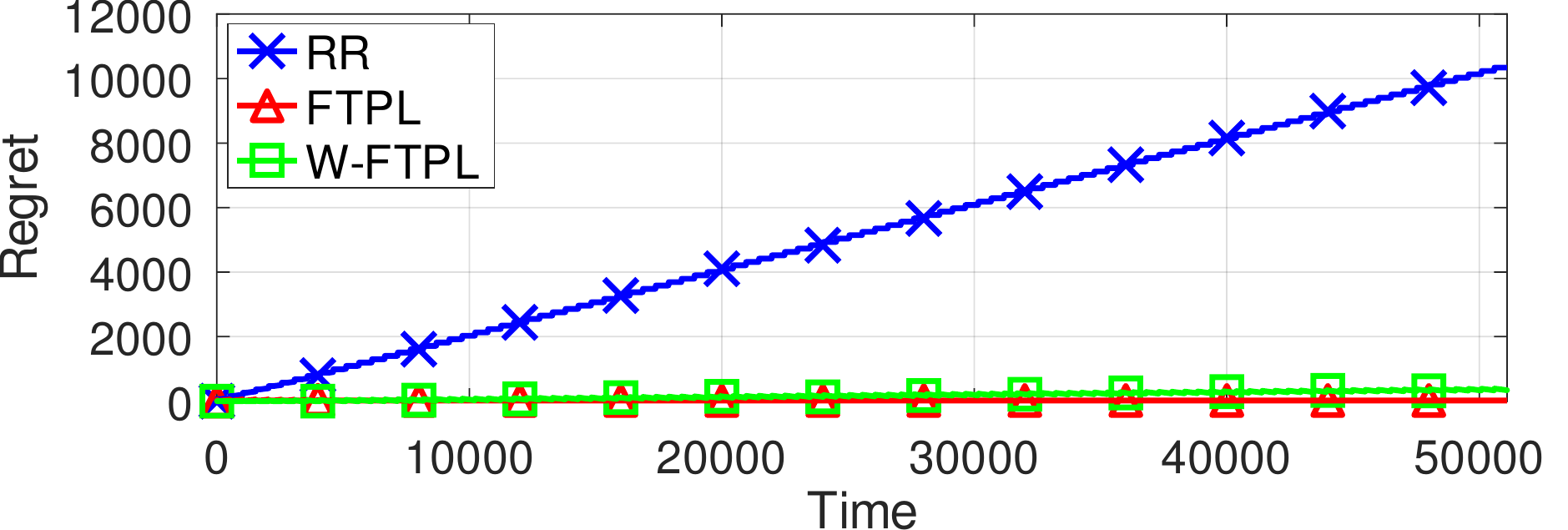}
	\caption{Regret as a function of  Time horizon ($T$)}
	\label{fig:RR_worst}
\end{figure}

\subsection{Stochastic arrivals}
 For the stochastic case, we consider arrivals in each slot to be Bernoulli with parameter $\mu$. We will restrict attention to the case where there is only one partial level of hosting, i.e., $\rho_{t}^\mathcal{P}\in \{0,\alpha_{2},1\}$. We consider $\alpha_{2}=0.5$ and $g(\alpha_{2})=0.45$. In \figurename \ref{fig:reg_T} and \figurename \ref{fig:reg_M}, we consider $\mu=0.4$, $c=0.45$.  All the results plotted are averaged over 50 independent experiments. In \figurename \ref{fig:reg_T},  we fix $M=5$ and compare the regret performance of the policies as a function of time. We observe that W-FTPL performs better than other policies, and RR has linear regret, which agrees with Theorem \ref{thm:reg_stch}. In \figurename \ref{fig:reg_M}, we compare the performance of the policies with respect to $M$ for $T=10^4$. We observe that W-FTPL performs better than other policies, and the main difference between the cost of policies is due to fetch cost. Also note that W-FTPL has better dependence on $M$ as compared to FTPL, which agrees with Theorem \ref{thm:reg_stch}. From Theorem \ref{thm:reg_stch} we observe that for large values of $M$ the regret of FTPL is proportional to $M/\Delta_{\text{min}}^2$ where as W-FTPL is proportional to $\log M /\Delta_{\text{min}}^2$. Therefore in \figurename \ref{fig:reg_c}, we compare the performance of the policies by varying rent cost $c$ for $M=500$, $\mu = 0.4$, and $T=10^4$. Again we note that W-FTPL outperforms RR and FTPL. As $c$ gets closer to $\mu$, the FTPL policy does multiple fetches, and as a result, the regret increases when $c$ is close to $\mu$.

\begin{figure}[ht]
	\centering
	\includegraphics[width=\linewidth]{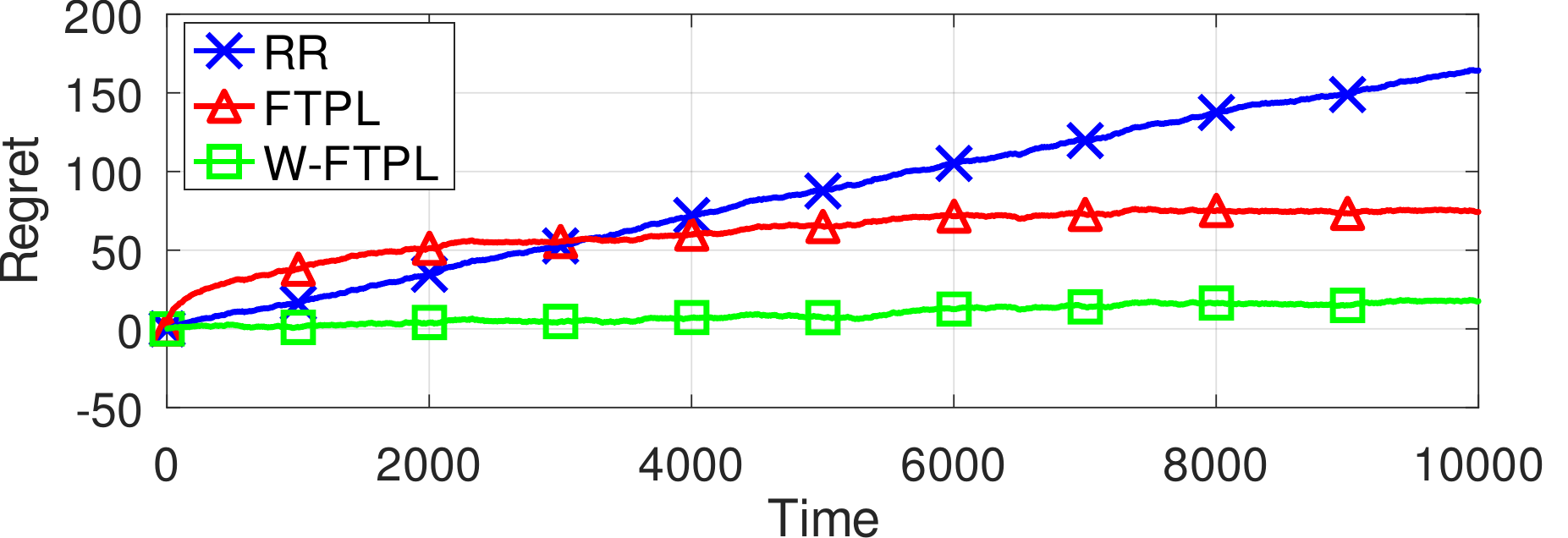}
	\caption{Regret as a function of time horizon ($T$)}
	\label{fig:reg_T}
\end{figure}

\begin{figure}[ht]
	\centering
	\includegraphics[width=\linewidth]{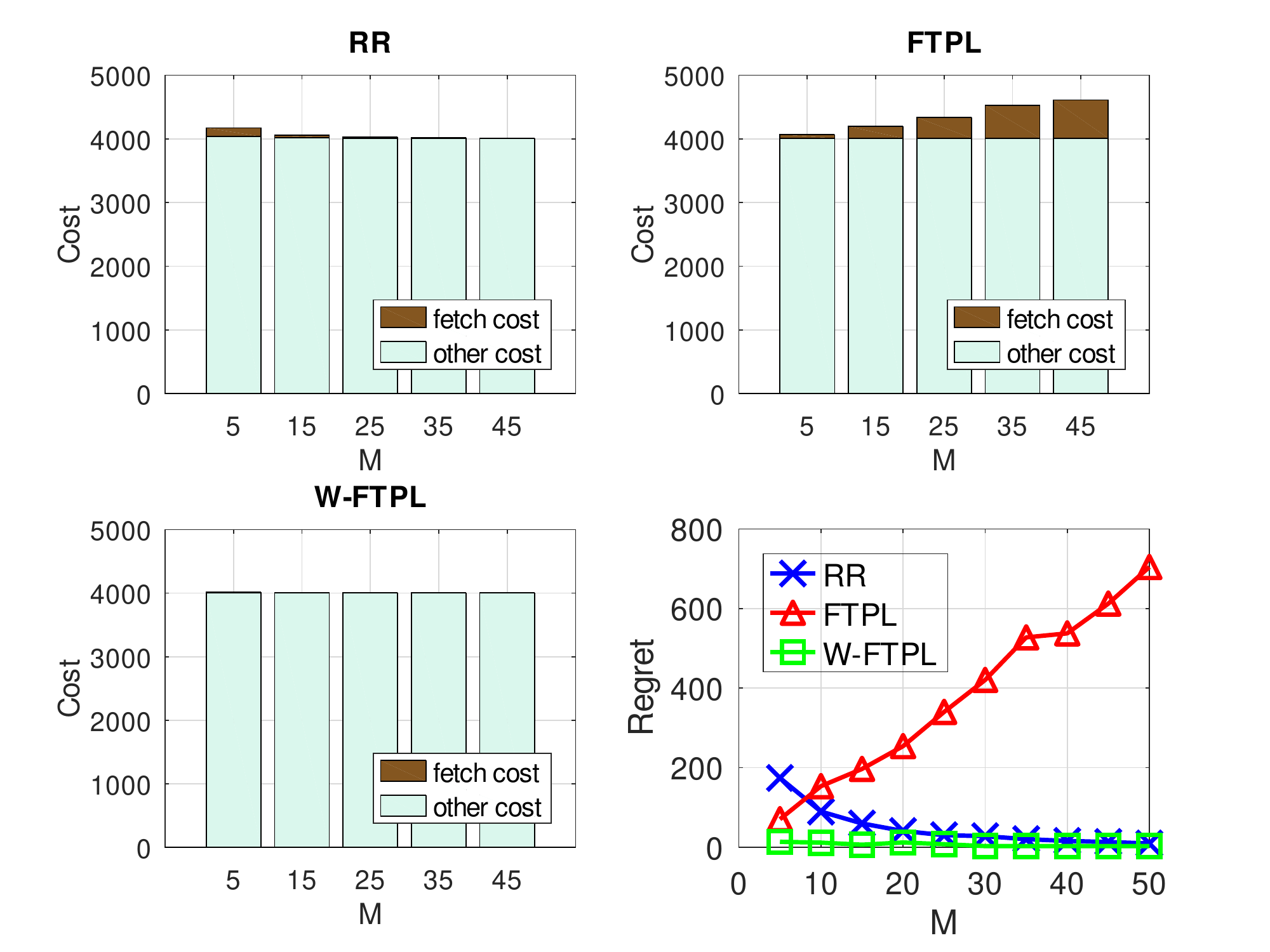}
	\caption{Cost, Regret as a function of fetch cost ($M$)}
	\label{fig:reg_M}
\end{figure}

\begin{figure}[ht]
	\centering
	\includegraphics[width=\linewidth]{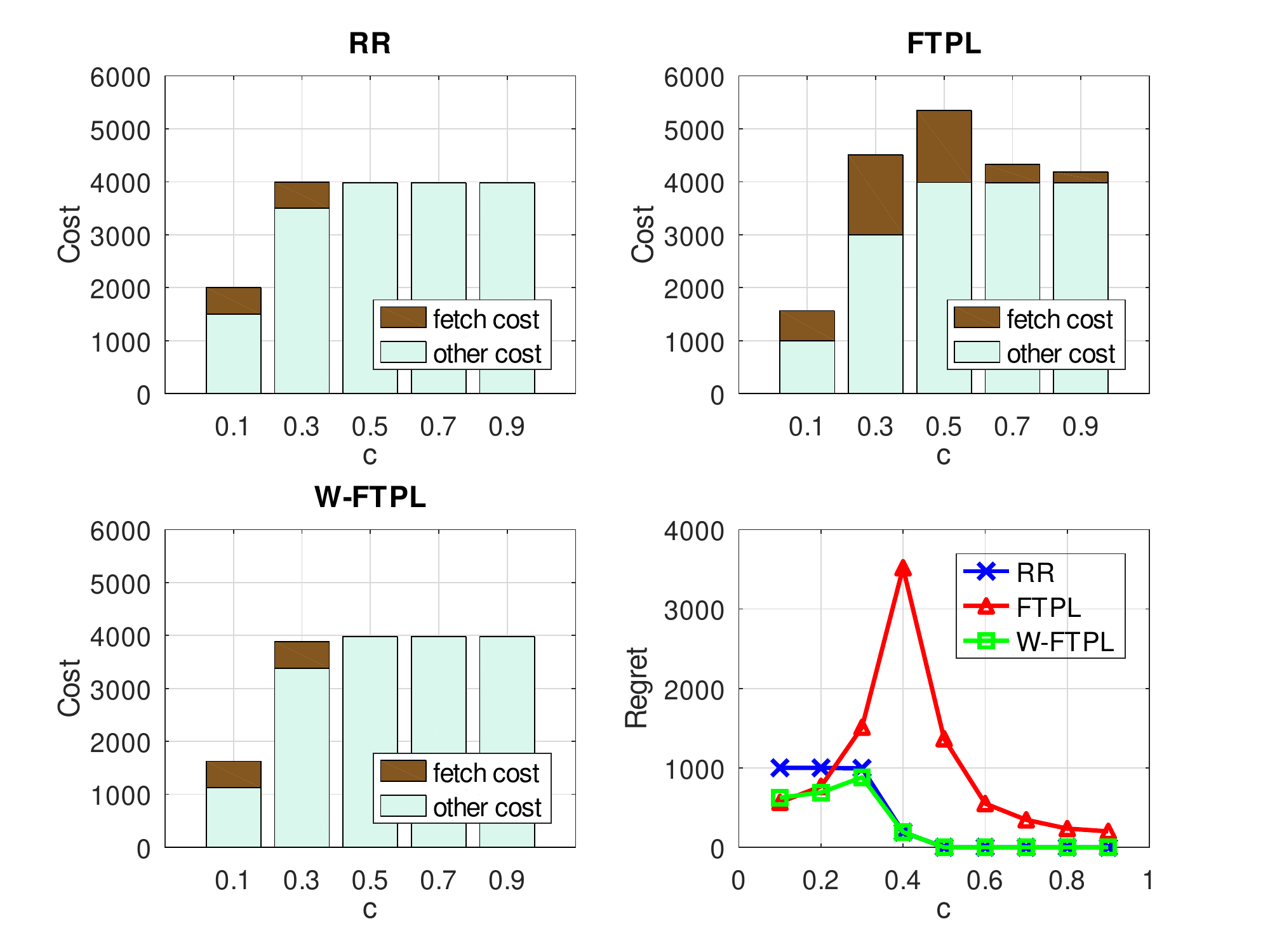}
	\caption{Cost, Regret as a function of  rent cost ($c$)}
	\label{fig:reg_c}
\end{figure}

\subsection{Trace driven arrivals}
For trace-driven simulations, we use the data collected by Grid Workloads Archive \cite{iosup2008grid}, which consists of requests arriving at computational servers. Among the traces mentioned in \cite{iosup2008grid} we use DAS-2 traces, which were provided by the Advanced School for Computing and Imaging (ASCI). We consider slot duration as 1 hour and plot a snapshot of  the trace in \figurename~\ref{fig:trace_arrvl}. We consider $\kappa=300$ for all the experiments in this subsection. All the results plotted for trace driven experiments are averaged over 50 independent experiments. 

We consider partial service hosting with one partial level for our simulations. This dataset provides information on the number of processors required to serve each request, and this number lies between 1 and 128 for all requests. For this experiment, we consider hosting the entire service as equivalent to having 128 processors at the edge and hosting the partial service as having 16 processors at the edge.  It follows that the hosting status $\rho_t^{\mathcal{P}} \in \{0, 16/125(=0.125), 1\}$. To compute $g(0.125)$, we use the dataset to compute the fraction of requests that require more than 16 processors for service and obtain that $g(0.125)=0.0328$.


\begin{figure}[ht]
	\centering
	\includegraphics[width=\linewidth]{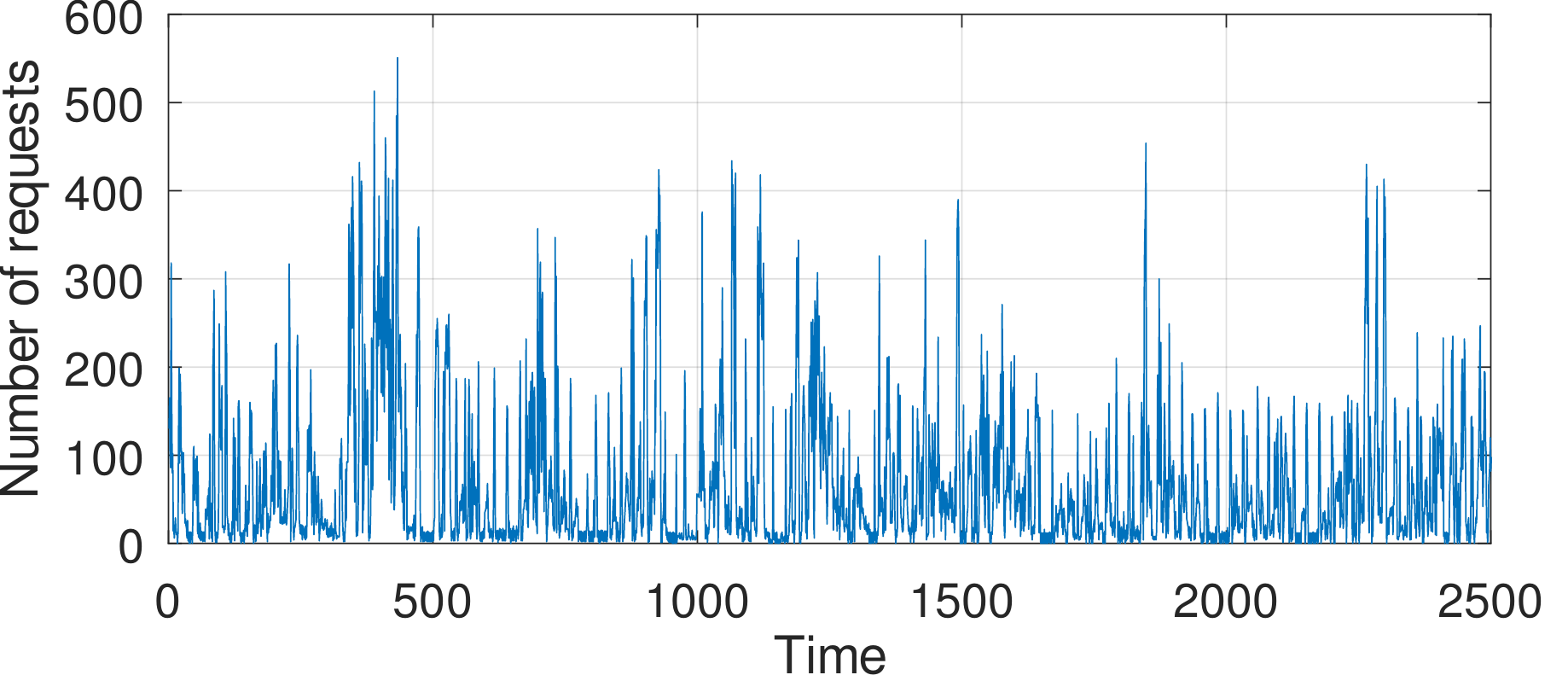}
	\caption{Arrivals in a slot for trace data considered }
	\label{fig:trace_arrvl}
\end{figure}

\begin{figure}[ht]
	\centering
	\includegraphics[width=\linewidth]{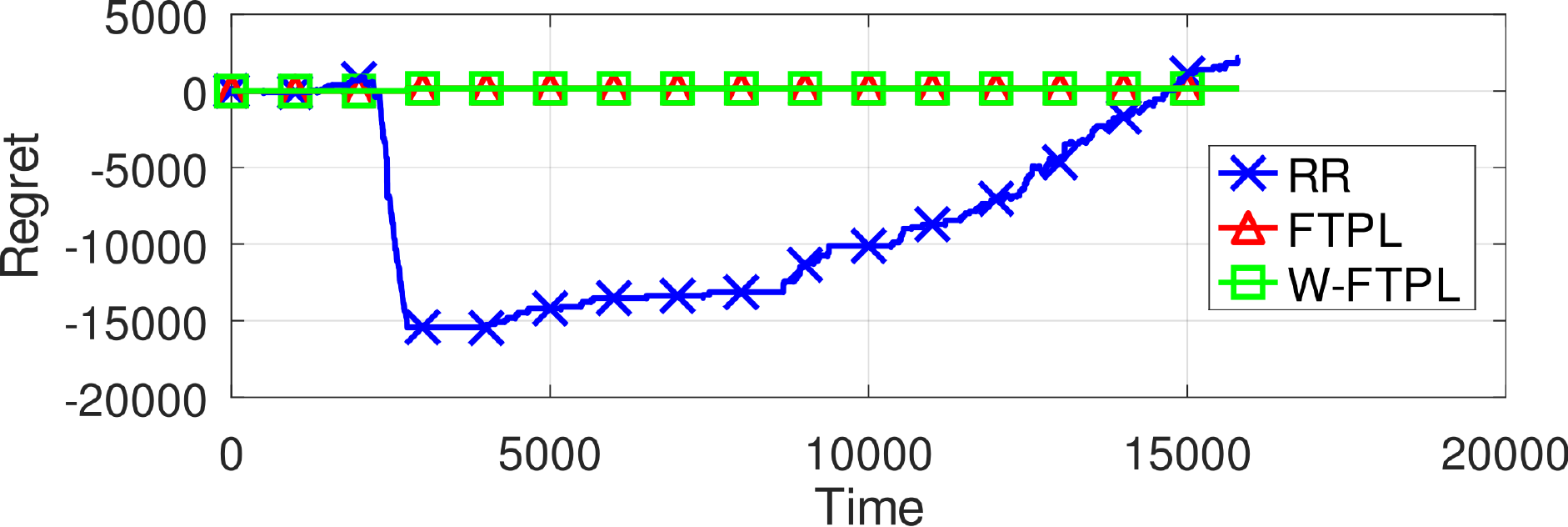}
	\caption{Regret as a function of time horizon}
	\label{fig:reg_alpha_T}
\end{figure}

\begin{figure}[ht]
	\centering
	\includegraphics[width=\linewidth]{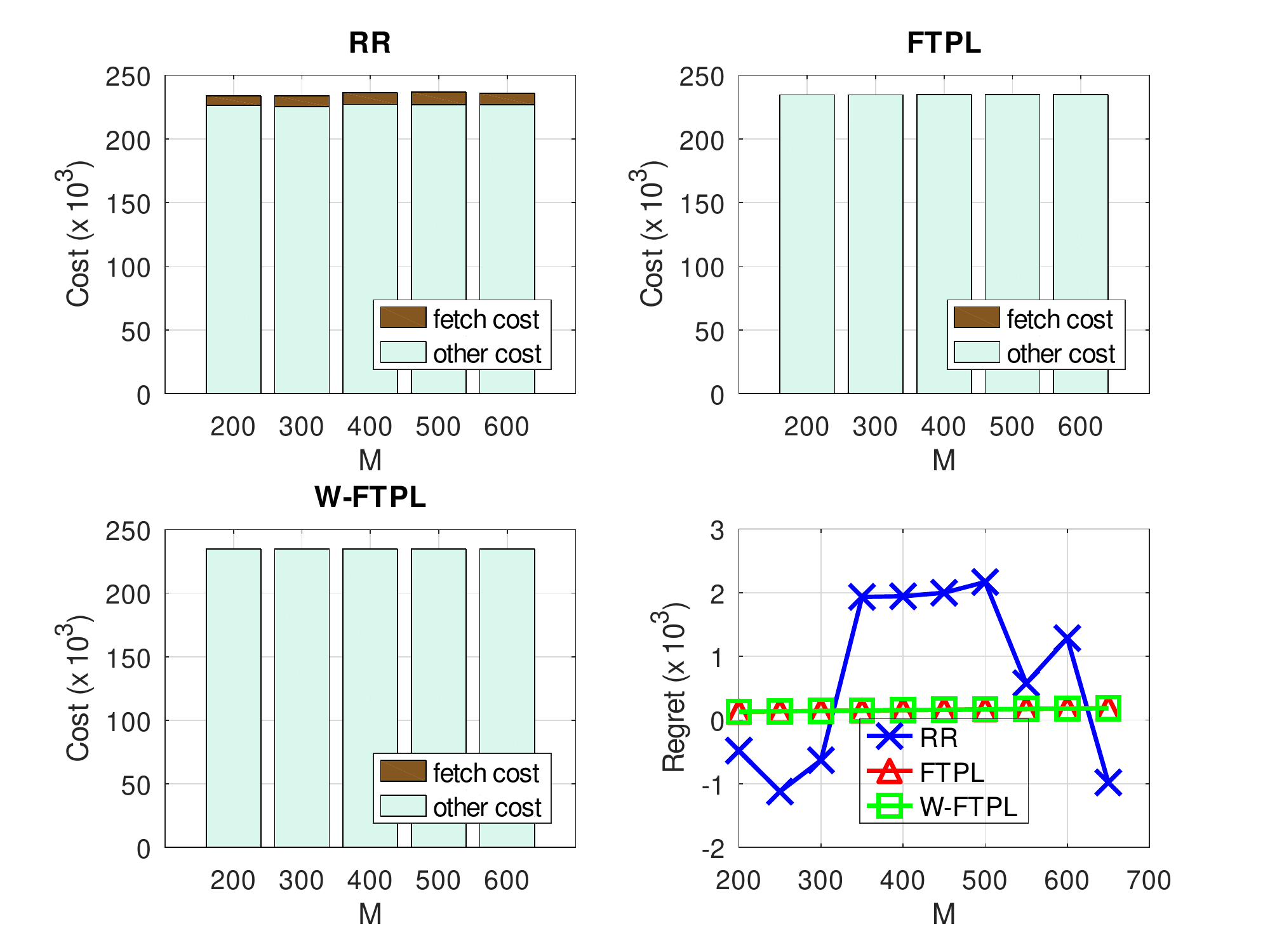}
	\caption{Cost, Regret as a function of fetch cost ($M$)}
	\label{fig:reg_M_adv}
\end{figure}

\begin{figure}[ht]
	\centering
	\includegraphics[width=\linewidth]{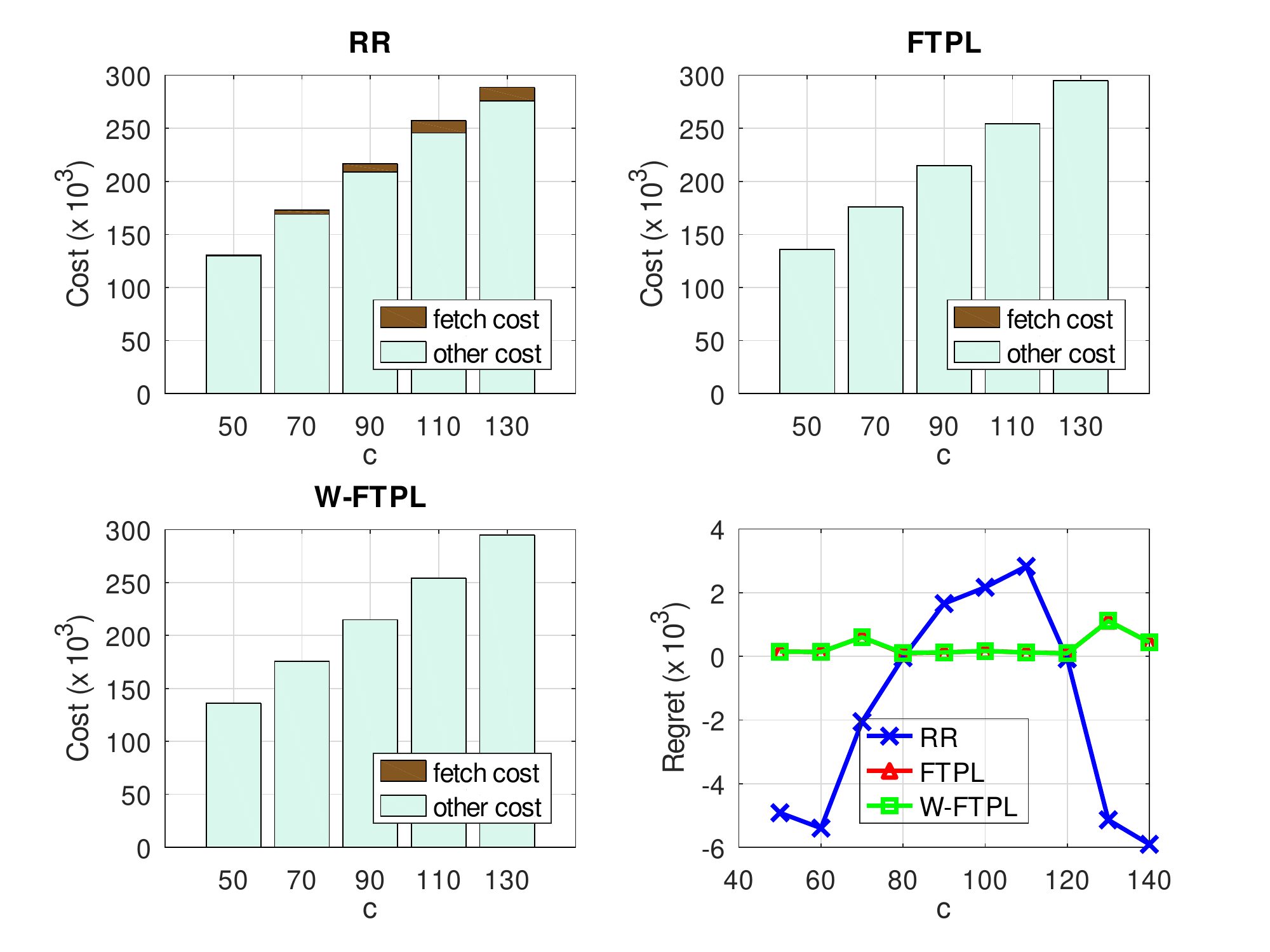}
	\caption{Cost, Regret as a function of  rent cost ($c$)}
	\label{fig:reg_c_adv}
\end{figure}

In \figurename \ref{fig:reg_alpha_T}, we consider $M=500$, $c=100$, $\kappa=300$ and compare the regret performance of policies as a function of time. We observe that RR performs better than other policies for most values of time horizon. This is because RR attempts to mimic the optimal offline policy, which is not a static policy for this arrival sequence. Note that FTPL and W-FTPL attempt to mimic the optimal static policy and perform worse than RR. 

 In \figurename \ref{fig:reg_M_adv}, we compare the performance of different policies with respect to $M$ and fix $c=100$, $\kappa=300$. We observe that RR performs better than the other policies for some values of $M$ because of the same reason mentioned before. In \figurename \ref{fig:reg_c_adv}, we compare the performance of different policies by varying rent cost $c$ and fix $M=500$, $\kappa=300$. Again we note that RR outperforms W-FTPL and FTPL in some cases.

Our simulation results thus indicate that FTPL and W-FTPL outperform RR for i.i.d. stochastic arrivals. In the case where the arrival process is such that the offline optimal policy is not static, RR can outperform FTPL and W-FTPL. Additionally, RR can have linear regret with respect to time, while FTPL and W-FTPL have sub-linear regret with respect to time for all arrival processes considered. 
\section{Proof of theorems}
\label{sec:proofs}
\subsection{Proof of Theorem \ref{thm:reg_adv}(a)}
\begin{proof}[Proof of Theorem \ref{thm:reg_adv}(a)]
Recall that regret of any policy is given by
	\begin{align*}
		\mathcal{R}^{\mathcal{P}}_{A}(T,r)=&\sum_{t=1}^Tc\rho_t+ g(\rho_t) r_t+ M (\rho_t-\rho_{t-1})^+\\
		&- \min_i\left\{c\alpha_iT+g(\alpha_i)R_T+M\alpha_{i}\right\}.
	\end{align*}
		Note that $\sup_{r} \mathcal{R}^{\mathcal{P}}_A(T,r) \ge \E_r[\mathcal{R}^{\mathcal{P}}_{A}(T,r)]$  where the expectation is w.r.t some specified distribution over request sequences. We lower bound $\E_r[\mathcal{R}^{\mathcal{P}}_{A}(T,r)]$ to get a lower bound on $\mathcal{R}^{\mathcal{P}}_A(T)$.
	\begin{align*}
		\E_r[\mathcal{R}^{\mathcal{P}}_{A}(T,r)]=&\sum_{t=1}^T c\rho_t+ g(\rho_t)\E_r[r_t] + M (\rho_t-\rho_{t-1})^+\\
		& - \E_r [\min_i\left\{c\alpha_iT+g(\alpha_i)R_T+M\alpha_{i}\right\}]\\
		\ge &\sum_{t=1}^T c\rho_t+ g(\rho_t)\E_r[r_t] \\
		& - \E_r [\min_i\left\{c\alpha_iT+g(\alpha_i)R_T\right\}]-M.
	\end{align*}
	 Let $\ell=\argmin_{i\ne 1} \frac{\alpha_{i}}{1-g(\alpha_{i})}$, $X\sim \text{Ber}(\frac{c\alpha_{\ell}}{\kappa(1-g(\alpha_{\ell}))})$, and  $r_t =\kappa X$ and i.i.d overtime, therefore $\E_r[r_t]=\frac{c\alpha_{\ell}}{1-g(\alpha_{\ell})}$. From Assumption \ref{assump:g(f)_convex} we have $\alpha_i+g(\alpha_{i})r_t\le \kappa$ for all $1\le i\le K$ and $0\le r_t\le \kappa$ therefore $\frac{c\alpha_{\ell}}{\kappa(1-g(\alpha_{\ell}))}\le1$ and $X$ is a valid Bernoulli random variable. 
	\begin{align*}
		\E_r[\mathcal{R}^{\mathcal{P}}_{A}(T,r)]
		\ge& \sum_{t=1}^T c\rho_t+ g(\rho_t)\frac{c\alpha_{\ell}}{1-g(\alpha_{\ell})}\\
		& - \E_r [\min_i\left\{c\alpha_iT+g(\alpha_i)R_T\right\}]-M\\		
		= & \sum_{t=1}^T (1-g(\rho_t))\frac{c\rho_{t}}{(1-g(\rho_t))}\mathds{1}_{\rho_{t}\ne 0}+ g(\rho_t) \frac{c\alpha_{\ell}}{1-g(\alpha_{\ell})}\\
		& - \E_r [\min_i\left\{c\alpha_iT+g(\alpha_i)R_T\right\}]-M\\
		\overset{(a)}{\ge} & \sum_{t=1}^T (1-g(\rho_t))\frac{c\alpha_{\ell}}{(1-g(\alpha_\ell))}+ g(\rho_t) \frac{c\alpha_{\ell}}{1-g(\alpha_{\ell})}\\
		& - \E_r [\min_i\left\{c\alpha_iT+g(\alpha_i)R_T\right\}]-M\\
		\overset{(b)}{\ge} & \frac{cT\alpha_{\ell}}{1-g(\alpha_{\ell})}- \E_r[\min\{R_T,cT\alpha_{\ell}+g(\alpha_{\ell})R_T\}]-M\\
		\overset{(c)}{=} & \frac{1}{2}\E_r[|R_T-cT\alpha_{\ell}-g(\alpha_{\ell})R_T|]-M\\
		=& \frac{\kappa(1-g(\alpha_\ell))}{2}\E_r\left[\left|\frac{R_T}{\kappa} - \frac{Tc\alpha_{\ell}}{\kappa(1-g(\alpha_\ell))}\right|\right] - M\\
		\overset{(d)}{=}&\Omega(\sqrt{T}),
	\end{align*}		
	where $(a)$ follows from definition of $\ell$, $(b)$ follows because of $\min_{1\le k\le K} a_k\le \min\{a_i,a_j\}$ for $1\le i,j\le K$, $(c)$ follows from Lemma \ref{lem:min_E} in Section \ref{sec:appendix}, and $(d)$ follows from Lemma \ref{lem:binomial_t} in Section \ref{sec:appendix}. 
\end{proof}

\subsection{Proof of Theorem \ref{thm:reg_adv}(b)}
Note that the RR algorithm in \cite{narayana2021renting} can have  only one intermediate level of partial hosting i.e., $\rho_{t}^{\text{RR}}\in \{0,\alpha_{2},1\}$. So throughout the is subsection we consider $\rho_{t}^{\text{RR}}\in \{0,\alpha_{2},1\}$.  We use the following lemmas to prove Theorem \ref{thm:reg_adv}(b).
\begin{lemma}\label{lem:RR_min_evict}
	RR does not fetch any non zero fraction of service for at least $\left\lceil  \frac{M\alpha_{i'}}{\kappa-c\alpha_{i^\prime} - g(\alpha_{i^\prime})\kappa}\right \rceil$ slots after eviction of complete service, where \\$\alpha_{i'} = \argmin_{\alpha_i\ne 0}\frac{M\alpha_{i}}{\kappa-c\alpha_{i}-\kappa g(\alpha_i)}$.
\end{lemma}

\begin{proof}
	Without loss of generality we consider $t=0$ is when RR last evicted the service. If RR policy fetches $\alpha_{i}\ne0$ fraction of service in slot $t+1$, then $c\alpha_{i} t+g(\alpha_{i})R_{t}+M\alpha_{i} < R_{t}$. 
	
	We now prove the lemma by contradiction. Let us assume that RR fetches  $\alpha_{i}$ fraction of the service in time slot $t_f+1$ and $t_f< \left\lceil \frac{M\alpha_{i'}}{\kappa-c\alpha_{i'} - g(\alpha_{i'})\kappa}\right\rceil.$ Then  by using the definition of $\alpha_{i'}$, we get $t_f<\frac{M\alpha_{i}}{\kappa-c\alpha_{i} - g(\alpha_{i})\kappa}$ and $c\alpha_{i} t_f+g(\alpha_{i})\kappa{t_f}+M\alpha_{i}>\kappa t_f\ge R_{t_f}$ which is a contradiction.
\end{proof}

\begin{lemma}\label{lem:RR_min_fetch}
	Once RR fetches any non zero fraction of service, it hosts (some positive fraction of the service) for at least $\lceil\frac{M}{c}\rceil$ slots before it evicts the completely.
\end{lemma}

\begin{proof}
	Without loss of generality we consider $t=0$ is when RR last fetched the service and say it hosted $\alpha_i\ne 0 $ fraction of the service. For the RR policy to evict the service in slot $t+1$, we need $ c\alpha_{i}t + g(\alpha_{i})R_{t}>  M\alpha_{i}+R_{t}$, or equivalently $R_{t}< \frac{c\alpha_i t-M\alpha_i}{1-g(\alpha_{i})}$. 
	
	We now prove the lemma by contradiction. Let us assume that RR evicts the service in time slot $t_e+1$ and $t_e<\lceil \frac{M}{c}\rceil$. Then from the assumption above,  we must have $R_{t_e}<\frac{c\alpha_i t_e-M\alpha_i}{1-g(\alpha_i)}<0$ which is a contradiction. 
\end{proof}
	
\begin{proof}[Proof of Theorem \ref{thm:reg_adv}(b)]
We split the entire time duration $T$ into frames of size $t_f+t_e$ and consider the  arrival sequence in each frame as  $r_t=\kappa$ for $1 \le t \le t_f$ and $r_t=0$ for $t_f+1\le t \le t_f+t_e$, where $\alpha_{i'} = \argmin_{\alpha_i\ne 0}\frac{M\alpha_{i}}{\kappa-c\alpha_{i}-\kappa g(\alpha_i)}$, $t_f=\left\lceil \frac{M\alpha_{i'}}{\kappa-c\alpha_{i'} - g(\alpha_{i'})\kappa} \right\rceil$,  and $t_e=\lceil\frac{M}{c}\rceil$. 

From Lemma \ref{lem:RR_min_evict}  we know that RR does not fetch any fraction of service till time $t_f$. At $t=t_f+1$ RR fetches some positive fraction of the service because $c\alpha_{i'}t_f + g(\alpha_{i'})R_{t_f} + M\alpha_{i'} < R_{t_f}$ holds by definition of $t_f$.
 
From Lemma \ref{lem:RR_min_fetch} we know that RR does not evict the service for $t\in[t_f+1, t_f+t_e]$. At $t=t_f+t_e+1$ RR evicts the  service because for $\alpha_i\ne 0$,  $c\alpha_{i}t_e+g(\alpha_{i})\sum_{t=t_f+1}^{t_f+t_e}r_t<M\alpha_{i}+\sum_{t=t_f+1}^{t_f+t_e}r_t$ holds by definition of $t_e$.

So amongst the first $t_f + t_e$ time-slots, the RR policy does not host for the first $t_f$ slots and then hosts some positive fraction of service for the following $t_e$ slots. Let $i^*=\argmin_i c\alpha_iT+g(\alpha_{i}R_T+M\alpha_i)$ be the fraction of service to be hosted  by optimal static policy. If the optimal static policy is to host $\alpha_{i^*}\ne 0$ fraction of service then we have a regret of at least $(\kappa -c\alpha_{i^*} -g(\alpha_{i^*})\kappa) t_f-M\alpha_{i^*}$ for initial frame and $(\kappa -c\alpha_{i^*} -g(\alpha_{i^*})\kappa) t_f=(\kappa -c\alpha_{i^*} -g(\alpha_{i^*})\kappa) \lceil \frac{M\alpha_{i'}}{\kappa-c\alpha_{i'} - g(\alpha_{i'})\kappa} \rceil>0$ form second frame onwards because of forwarding, which is a constant and does not depend on $T$. If the optimal static policy is not to host the service, then we have a regret of $c\alpha_{i}t_e>0$ till $t_f+t_e$ for some $\alpha_{i}\ne 0$, which is a constant and does not depend on $T$. So in either case, a regret larger than some positive constant(say $d$) is occurred in a frame of size $t_f+t_e$. Since we split the entire time duration $T$ into frames of size $t_f+t_e=\lceil \frac{M\alpha_{i'}}{\kappa-c\alpha_{i^\prime}-g(\alpha_{i^\prime})\kappa}\rceil + \lceil\frac{M}{c} \rceil$ and the request sequence is repeated in each frame, we get linear regret.

\begin{align*}
	\mathcal{R}^{\text{RR}}_A(T)&\ge \sum_{f=1}^{\lfloor T/(t_f+t_e) \rfloor} d	- M\alpha_{i^*}\\
	&\ge \lfloor T/(t_f+t_e) \rfloor d- M\alpha_{i^*}\\
	&\ge d\left(\frac{T}{\frac{M\alpha_{i'}}{\kappa-c\alpha_{i^\prime}-g(\alpha_{i^\prime})\kappa}+ \frac{M}{c} +2}-1\right)- M\alpha_{i^*} \\
	&=\Omega(T).
\end{align*}	
\end{proof}

\subsection{Proof of Theorem \ref{thm:reg_adv}(c)}
  The optimal static policy only fetches the optimal fraction of service to be hosted once. Therefore,
\begin{align}
\mathcal{R}^{\text{FTPL}}_{A}(T,r)=&\E[\mathcal{C}^\text{FTPL}(T,r)] \nonumber\\
&- \min_i\{c\alpha_iT+g(\alpha_i)R_T+M\alpha_{i}\} \nonumber\\
\le& \E[\mathcal{C}^\text{FTPL}(T,r)] - \min_i\{c\alpha_iT+g(\alpha_i)R_T\}.
 \label{eq:FTPL_reg}
\end{align}
Since FTPL policy does not consider fetch cost $M$ while making decisions we can decouple the fetch cost and the non fetch cost incurred by the FTPL policy. Therefore we can also decouple the expected regret incurred by FTPL into expected regret without fetch cost ($M=0$) and the fetch cost incurred by FTPL to bound \eqref{eq:FTPL_reg}.
We first bound the expected regret of the FTPL policy with fetch cost $M=0$ and then add it's expected fetch cost to get the final regret bound.
\begin{lemma}
	The regret for 	FTPL policy  with non decreasing learning rate $\{\eta_t\}_{t=1}^T$ and $M=0$ is given by
	\begin{align*}
	\mathcal{R}^{\text{FTPL}}_{A}(T)&\le \sqrt{2\log K}\left( \eta_T  + \kappa^2 \sum_{t=1}^{T}\frac{1}{\eta_t}\right) .
	\end{align*}
	\label{lem:FTPL_reg_adv}
\end{lemma}
\begin{proof}
	The proof of this theorem  follows along the same lines as Theorem 1 of \cite{cohen2015following} and, Proposition 4.1 of \cite{mukhopadhyay2021online}. Recall that $\rho_{t}^{\mathcal{P}}$ represents the fraction of service hosted in slot $t$ by policy $\mathcal{P}$ and $\bm{\rho}_t^\mathcal{P}$ is a one hot vector where position of one represents the level of service hosted at the edge.
	In slot $t$ FTPL hosts fraction of service $\rho_{t}^{\text{FTPL}}=\langle\bm{\rho}_{t}^{\text{FTPL}},\bm{s} \rangle$, where $\bm{\rho}_t^{\text{FTPL}}= \argmin_{\bm{\rho}\in \mathcal{X}} \langle \bm{\rho},\bm{\Theta}_t+\eta_t \bm{\gamma}\rangle$.
	For ease of notation we suppress the policy in super script and denote $\bm{\rho}_t^{\mathcal{P}}$ as $\bm{\rho}_t$ and $\rho_{t}^{\mathcal{P}}$ as $\rho_{t}$ in further discussion. Define a time varying potential function 
	\begin{align}
	\bm{\Phi}_t(\bm{\theta})= \E_{\bm{\gamma}} \left[ \min_{\bm{\rho}\in \mathcal{X}} \langle \bm{\rho},\bm{\theta} + \eta_t \bm{\gamma} \rangle \right], 
	\end{align}
	 and observe that the gradient of $\bm{\Phi}_t$ at $\bm{\Theta}_t$ is $\E_{\bm{\gamma}}[\bm{\rho}_t]$. So $\langle \nabla\bm{\Phi}_t(\bm{\Theta}_t), \bm{\theta}_t \rangle =\E_{\bm{\gamma}}[\langle \bm{\rho}_t, \bm{\theta}_t \rangle]$. Consequently
	\begin{align*}
	\E_{\bm{\gamma}}[\langle \bm{\rho}_t, \bm{\theta}_t \rangle] &= \langle \nabla\bm{\Phi}_t(\bm{\Theta}_t),\bm{\Theta}_{t+1}- \bm{\Theta}_t \rangle, \\
	&= \bm{\Phi}_t(\bm{\Theta}_{t+1})- \bm{\Phi}_t(\bm{\Theta}_t)-\frac{1}{2} \langle \bm{\theta}_t, \nabla^2\Phi_t(\tilde{\bm{\theta}}_t) \bm{\theta}_t) \rangle,
	\end{align*}
	where $\tilde{\bm{\theta}}_t$ is the line segment joining $\bm{\Theta_{t+1}}$ and $\bm{\Theta_t}$ which follows from Taylor's expansion of $\bm{\Phi}_t(\bm{\Theta}_{t+1})$. 
	\begin{align*}
	&\sum_{t=1}^{T}\E_{\bm{\gamma}}\left[ \langle \bm{\rho}_t, \bm{\theta}_t \rangle \right]\\
	&=\sum_{t=1}^T \bm{\Phi}_t(\bm{\Theta}_{t+1})- \bm{\Phi}_t(\bm{\Theta}_t)-\frac{1}{2} \langle \bm{\theta}_t, \nabla^2\Phi_t(\tilde{\bm{\theta}}_t) \bm{\theta}_t) \rangle,\\
	&=\bm{\Phi}_{T}(\bm{\Theta}_{T+1}) - \bm{\Phi}_{1}(\bm{\Theta}_{1}) +\sum_{t=2}^T \bm{\Phi}_{t-1}(\bm{\Theta}_{t})- \bm{\Phi}_t(\bm{\Theta}_t)\\
	&-\sum_{t=1}^T \frac{1}{2} \langle \bm{\theta}_t, \nabla^2\Phi_t(\tilde{\bm{\theta}}_t) \bm{\theta}_t \rangle. \numberthis \label{eq:reg_adv_thetas}
	\end{align*}
	By Jensen’s inequality 
	\begin{align*}
	\bm{\Phi}_{T}(\bm{\Theta}_{T+1}) & = \E \left[ \min_{\bm{\rho} \in \mathcal{X}}\langle \bm{\rho}, \bm{\Theta}_{T+1}+\eta_t \bm{\gamma} \rangle \right]\\
	&\le \min_{\bm{\rho} \in \mathcal{X}} \E \left[ \langle \bm{\rho}, \bm{\Theta}_{T+1}+\eta_t \bm{\gamma} \rangle \right]\\
	&\overset{(a)}{=} \min_{\bm{\rho} \in \mathcal{X}}\langle \bm{\rho}, \bm{\Theta}_{T+1} \rangle,
	\end{align*}
	where $(a)$ follows because $\gamma_i$'s are standard Gaussian random variables. Recall that $\min_{\bm{\rho} \in \mathcal{X}}\langle \bm{\rho}, \bm{\Theta}_{T+1} \rangle$ is the loss incurred by the offline static policy. Therefore by rearranging \eqref{eq:reg_adv_thetas} and using Jensen’s inequality we get
	\begin{align*}
	\mathcal{R}^{\text{FTPL}}_A(T) \le& -\bm{\Phi}_{1}(\bm{\Theta_{1}}) + \sum_{t=2}^T \bm{\Phi}_{t-1}(\bm{\Theta}_{t})- \bm{\Phi}_t(\bm{\Theta}_t)\\
	&-\sum_{t=1}^T \frac{1}{2} \langle \bm{\theta}_t, \nabla^2\Phi_t(\tilde{\bm{\theta}}_t) \bm{\theta}_t \rangle. \numberthis \label{eq:Reg_ftpl_adv}
	\end{align*} 
	We bound each term in the RHS separately to get the final result. 
	
	The first term is $-\bm{\Phi}_{1}(\bm{\Theta}_{1})=-\bm{\Phi}_{1}(0)=-\eta_1\E_{\bm{\gamma}}[\min_{\bm{\rho}} \langle \bm{\rho}, \bm{\gamma} \rangle]$\\$=\eta_1\E_{\bm{\gamma}}[\max_{\bm{\rho}\in \mathcal{X}} \langle \bm{\rho}, \bm{\gamma} \rangle] $ since Gaussian random variables are symmetric. By using the result of Lemma 9 in \cite{cohen2015following}, we get 
	\begin{equation}
	-\bm{\Phi}_{1}(0)\le \eta_1 \sqrt{2 \log K}. \label{eq:FTPL_reg_term1}
	\end{equation}
	
	The second term can be analyzed by considering 
	\begin{align*}
	&\bm{\Phi}_{t-1}(\bm{\Theta}_{t})- \bm{\Phi}_t(\bm{\Theta}_t) \nonumber\\
	&=\E_{\bm{\gamma}} [ \min_{\bm{\rho}\in \mathcal{X}} \langle \bm{\rho},\bm{\theta}_{t} + \eta_{t-1} \bm{\gamma} \rangle- \min_{\bm{\rho}\in \mathcal{X}} \langle \bm{\rho},\bm{\theta_t} + \eta_t \bm{\gamma} \rangle ]\\ 
	&\overset{(a)}{\le} \E_{\bm{\gamma}} [ \max_{\bm{\rho}\in \mathcal{X}} \langle \bm{\rho}, (\eta_{t-1}-\eta_{t}) \bm{\gamma} \rangle ] \\
	&\le |\eta_t-\eta_{t-1}|\sqrt{2\log K}, \numberthis \label{eq:FTPL_reg_term2}
	\end{align*} 
	 where $(a)$ follows because $\max_{\bm{y}\in \mathcal{Y}}(f(\bm{y})-g(\bm{y}) )\ge \max_{\bm{y}\in \mathcal{Y}}f(\bm{y})-\max_{\bm{y}\in \mathcal{Y}}g(\bm{y})$ for any arbitrary functions $f$, $g$. We  obtain \eqref{eq:FTPL_reg_term2} by using Lemma 9 in \cite{cohen2015following}.
	
	Now we bound the third term in \eqref{eq:Reg_ftpl_adv}. Let $H=\nabla^2\Phi_t(\tilde{\bm{\theta}}_t)$. By using Lemma 7 of \cite{abernethy2014online} we have,
	\begin{align*}
		H_{i,j}=\frac{1}{\eta_t} \E[\hat{x}(\tilde{\bm{\theta}_t}+\eta_t \bm{\gamma} )_i \bm{\gamma}_j].		
	\end{align*}
	If we replace $\eta_{t}$ with $\eta_{t}/\kappa$ and $\theta_{t}$ with $\theta_{t}/\kappa$ and apply Lemma 2 of \cite{cohen2015following} we get 
	\begin{align*}
		-\bigg\langle \frac{\bm{\theta}_t}{\kappa}, \kappa H \frac{\bm{\theta}_t}{\kappa} \bigg\rangle \le \frac{2\kappa}{\eta_{t}} \sqrt{2\log K} \\
		\implies -\langle \bm{\theta}_t, \nabla^2\Phi_t(\tilde{\bm{\theta}}_t) \bm{\theta}_t \rangle \le \frac{2\kappa^2}{\eta_{t}} \sqrt{2\log K}. \numberthis  \label{eq:FTPL_reg_term3}
	\end{align*}	
	By using \eqref{eq:FTPL_reg_term1}, \eqref{eq:FTPL_reg_term2} and \eqref{eq:FTPL_reg_term3}, we bound \eqref{eq:Reg_ftpl_adv}, 
	\begin{align*}
		\mathcal{R}^{\text{FTPL}}_A(T)\le & \eta_1 \sqrt{2\log K}+\sum_{t=2}^{T} |\eta_t-\eta_{t-1}|\sqrt{2\log K}\\
		&+\sum_{t=1}^{T} \frac{2\kappa^2}{\eta_t} \sqrt{2\log K}\\
		\le& \eta_T \sqrt{2\log K} + \sum_{t=1}^{T} \frac{2\kappa^2}{\eta_t} \sqrt{2\log K}\\
		\le&\eta_T \sqrt{2\log K} + 2\kappa^2 \sqrt{2\log K} \sum_{t=1}^{T}\frac{1}{\eta_t}.
	\end{align*}
\end{proof}

Now we bound the expected fetch cost incurred by FTPL policy.
\begin{lemma}\label{lem:FTPL_swcost_adv}
	The expected fetch cost under FTPL policy upto time $T$ denoted by $\E[\mathcal{C}^{\text{FTPL}}_{F}(T)]$ with learning rate $\eta_t=\alpha \sqrt{t}$ is bounded as follows	
	$$\E[\mathcal{C}^{\text{FTPL}}_{F}(T)] \le \frac{MK^2 (c+2\kappa)}{2\alpha\sqrt{\pi}}{\sqrt{T+1}}.$$	
\end{lemma}
\begin{proof} The proof of this theorem follows similar a line of arguments as in  \cite{mukhopadhyay2021online} for the case with $N=2$ and cache size 1 case.
	The fetch cost is incurred only when the service is fetched. The overall cost of fetching is $\sum_{t=1}^{T}M (\rho_{t+1}-\rho_{t})^+$. So, the expected fetch cost is given by  
	\begin{align*}
	\E[\mathcal{C}^{\text{FTPL}}_{F}(T)]&= \E\left[ \sum_{t=1}^{T}M (\rho_{t+1}-\rho_{t})\mathds{1}_{\{\rho_{t+1}>\rho_{t}\}} \right]\\
	&\le M\sum_{t=1}^T \Prob(\rho_{t+1}>\rho_{t})\\
	&\le  M\sum_{t=1}^T \sum_{j>i} \Prob(\rho_{t+1}=\alpha_{j},\rho_{t}=\alpha_{i}).\\
	\end{align*}
	
	Let $\mathcal{E}_{i,j}^{(t)}$ be the event that $\rho_{t}=\alpha_{i}$, $\rho_{t+1}=\alpha_{j}$. $\mathcal{E}_{i,j}^{(t)}$ occurs only if $\Theta_{t,i}+\eta_t \gamma_i< \Theta_{t,j}+\eta_t \gamma_j$, $\Theta_{t+1,i}+\eta_{t+1} \gamma_i> \Theta_{t+1,j}+\eta_{t+1} \gamma_j$. Let $\delta_\alpha=\alpha_j-\alpha_i$, $\delta_g=g(\alpha_i)-g(\alpha_j)$, note that $0<\delta_\alpha\le1$, $0<\delta_g\le 1$.
	\begin{align*}
	&\Prob(\mathcal{E}_{i,j}^{(t)})\\
	&\le\Prob\left(\Theta_{t,i}+\eta_{t} \gamma_i< \Theta_{t,j}+\eta_{t} \gamma_j,\right.\\
	&\quad\qquad \left.\Theta_{t+1,i}+\eta_{t+1} \gamma_i> \Theta_{t+1,j}+\eta_{t+1} \gamma_j\right)\\
	&= \Prob\left( \frac{\Theta_{t+1,j}-\Theta_{t+1,i}}{\sqrt{2}\eta_{t+1}} < \frac{(\gamma_i-\gamma_j)}{\sqrt{2}} <\frac{\Theta_{t,j}-\Theta_{t,i}}{\sqrt{2}\eta_{t}} \right)\\ 
	&\overset{(a)}{\le} \frac{\Theta_{t,j}-\Theta_{t,i}}{2\sqrt{\pi}\eta_{t}}- \frac{\Theta_{t+1,j}-\Theta_{t+1,i}}{2\sqrt{\pi}\eta_{t+1}}\\ 
	&=\frac{1}{2\sqrt{\pi}}\left( (c(t-1)\delta_\alpha- \delta_gR_{t-1})\left(\frac{1}{\eta_{t}} -\frac{1}{\eta_{t+1}}\right) \right)\\
	&\quad- \frac{1}{2\sqrt{\pi}} \left(\frac{c\delta_\alpha-\delta_g r_{t}}{\eta_{t+1}}\right)\\ 
	&\le \frac{1}{2\sqrt{\pi}} \left(c t\left(\frac{1}{\eta_{t}}-\frac{1}{\eta_{t+1}}\right)+\frac{\delta_g r_{t}-c\delta_\alpha}{\eta_{t+1}}\right)\\
	&\overset{(b)}{\le} \frac{1}{2\sqrt{\pi}} \left(c t \left(\frac{1}{\eta_{t}} - \frac{1}{\eta_{t+1}}\right) + \frac{\kappa}{\eta_{t+1}}\right) \numberthis \label{eq:FTPL_sw_eta}\\
	&=\frac{1}{2 \alpha\sqrt{\pi}} \left(c t\left( \frac{1}{\sqrt{t}}-\frac{1}{\sqrt{t+1}}\right)+\frac{\kappa}{\sqrt{t+1}}\right)\\
	&=\frac{1}{2 \alpha\sqrt{\pi}} \left( c\sqrt{\frac{t}{t+1}}  \frac{1}{\sqrt{t}+\sqrt{t+1}} +\frac{\kappa}{\sqrt{t+1}}\right)\\
	&\overset{(c)}{\le}\frac{(c+2\kappa)}{4 \alpha\sqrt{\pi(t+1)}}.
	\end{align*}
	 Here $(a)$ follows from the fact that $\Prob(a<Z<b)\le \frac{b-a}{\sqrt{2\pi}}$,  where $Z\sim\mathcal{N}(0,1)$, $(b)$ follows because $\delta_g\le 1$ and $r_t\le \kappa$ and, $(c)$ follows because $\frac{\sqrt{t}}{\sqrt{t}+\sqrt{t+1}}\le \frac{1}{2}$. 
	Therefore, 
	\begin{align*}
	\E[\mathcal{C}_f^{\text{FTPL}}(T)] &\le \frac{M(c+2\kappa)}{4\alpha\sqrt{\pi}} \sum_{t=1}^T \sum_{j>i} \frac{1}{\sqrt{t+1}}\\
	& \le \frac{M K^2(c+2\kappa)}{2\alpha\sqrt{\pi}}{\sqrt{T+1}}.
	\end{align*}	
\end{proof}
\begin{proof}[Proof of Theorem  \ref{thm:reg_adv}(c)]
	By using Lemma \ref{lem:FTPL_reg_adv} and \ref{lem:FTPL_swcost_adv} we get,
	\begin{align*}
	\mathcal{R}^{\text{FTPL}}_A(T)\le&\sqrt{2\log K}\left( \alpha\sqrt{T}  + 2\kappa^2 \sum_{t=1}^{T}\frac{1}{\alpha\sqrt{t}}\right) \\
	&+\frac{K^2M(c+2\kappa)}{2\alpha\sqrt{\pi}} \sqrt{T+1}\\
	\le& \sqrt{2T\log K}\left( \alpha  + \frac{4\kappa^2}{\alpha}\right)\\ &+\frac{K^2M(c+2\kappa)}{2\alpha\sqrt{\pi}} \sqrt{T+1}.
	\end{align*}
\end{proof}

\subsection{Proof of Theorem \ref{thm:reg_adv}(d)}
\begin{proof}[Proof of Theorem \ref{thm:reg_adv}(d)]
	We characterize the fetch cost under FTPL to analyze the impact of $M$ on the total cost incurred under FTPL. Fetch cost is incurred if $\rho_{t-1}=0$, $\rho_{t}=1$. Consider a request sequence $r_1=1$ and $r_t=0$ for all $t\ge2$. The optimal static policy is to not host any fraction of service i.e., to forward all the requests to back-end server. 
	
	
	The probability of hosting $\alpha_{i}$, $1\le i\le K$ fraction of service in first slot is  $1/K$. Hosting any non zero fraction of service will incur at least $M\alpha_{2}$ fetch cost. Therefore the expected fetch cost of FTPL policy will be at least $M\alpha_{2}/K$  which is a lower bound on the expected cost of FTPL policy for this request sequence. For FTPL we can decouple the fetch cost and the non fetch cost and  FTPL does not consider $M$ while making decisions. Therefore  $M\alpha_{2}/K$ is also a lower bound for regret of FTPL. This completes the proof.
	
\end{proof}

\subsection{Proof of Theorem \ref{thm:reg_adv}(e)}

\begin{proof}[Proof of Theorem \ref{thm:reg_adv}(e)]
	Recall that $T_s=\min\{t:t<\frac{(\max_{i\ne j}(\Theta_{t+1,i}-\Theta_{t+1,j}))^2}{\kappa^2\beta (\log M)^{1+\delta}}\}$. 	For arrival sequences with requests arriving in time-slots 1 to $T$, we have two possible cases, namely, $T_s\ge T$ and $T_s<T$. We consider each case to bound the regret of W-FTPL policy.
	
	\textit{Case I ($T_s\ge T$)}: Let $r^{(1)}$ be a request sequence chosen by adversary such that $T_s\ge T$, then 
	\begin{align*}
		|\max_{i\ne j}(\Theta_{T+1,j}-\Theta_{T+1,i})|&<\kappa\sqrt{\beta T (\log M)^{1+\delta}}\\
		\implies |\max_{i}(\Theta_{T+1,1}-\Theta_{T+1,i})|&<\kappa\sqrt{\beta T (\log M)^{1+\delta}}\\
		\overset{(a)}{\implies}\mathcal{R}^{\text{W-FTPL}}_A(T) &<\kappa\sqrt{\beta T (\log M)^{1+\delta}},
	\end{align*}
	$(a)$ follows because $\rho_{t}=0$ for  $t<T_s$ for W-FTPL.
	
	\textit{Case II ($T_s<T$)}: Let $r^{(2)}$ be a request sequence chosen by adversary such that $T_s<T$, then by using case I we can bound 
	\begin{align*}
	 |\max_{i\ne j}\Theta_{T_s,i}-\Theta_{T_s,j}|&<\kappa\sqrt{\beta (T_s-1) (\log M)^{1+\delta}}\\
	\implies \mathcal{R}^{\text{W-FTPL}}_{A}(T_s-1)&<\kappa\sqrt{\beta T (\log M)^{1+\delta}}.
	\end{align*}
		
	Thus combining both cases we get  upper bound on regret in wait phase as
	\begin{align}
	\mathcal{R}^{\text{W-FTPL}}_{A}(T_s-1)&<\kappa\sqrt{\beta T (\log M)^{1+\delta}}. \label{eq:WFTPL_wait}
	\end{align}
	Note that W-FTPL policy follows FTPL policy after its waiting time i.e., from time $T_s$. Therefore W-FTPL and FTPL take same decisions and have same cost from time $T_s$ to $T$. Thus the regret in a slot is also same for W-FTPL and FTPL from time $T_s$ to $T$. If $T_s\ge T$ then only wait phase will be there and regret in FTPL phase is considered to be zero. We denote regret from time $t_1$ to $t_2$ as $\mathcal{R}_A^{\mathcal{P}}(t_1:t_2)$.
	\begin{align*}
		\mathcal{R}_A^{\text{W-FTPL}}(T)&=\mathcal{R}_A^{\text{W-FTPL}}(T_s-1) +	\mathcal{R}_A^{\text{W-FTPL}}(T_s:T)\\
		&=\mathcal{R}_A^{\text{W-FTPL}}(T_s-1)+\mathcal{R}_A^{\text{FTPL}}(T_s:T)\\
		&\le \mathcal{R}_A^{\text{W-FTPL}}(T_s-1)+\mathcal{R}_A^{\text{FTPL}}(1:T). \numberthis \label{eq:ftpl_wftpl_adv_reg}
	\end{align*}

	By \eqref{eq:WFTPL_wait}, \eqref{eq:ftpl_wftpl_adv_reg} and Theorem \ref{thm:reg_adv}(c)  we get
		 
	\begin{align*}
	\mathcal{R}^{\text{W-FTPL}}_A(T)\le& \kappa\sqrt{\beta T (\log M)^{1+\delta}}\\
	&+\sqrt{2T\log K} \left( \alpha+ \frac{4\kappa^2}{\alpha^2}\right)\\ &+\frac{K^2M(c+2\kappa)}{2\alpha\sqrt{\pi}} \sqrt{T+1}.
	\end{align*}	
\end{proof}

\subsection{Proof of Theorem \ref{thm:reg_stch}(a)}

\begin{proof}[Proof of Theorem \ref{thm:reg_stch}(a)]
	From Lemma \ref{lem:RR_min_fetch} we know that once RR fetches any non-zero fraction of service, it hosts (some positive fraction of the service) for at least $\left \lceil \frac{M}{c}\right \rceil$ slots before it evicts completely. From Lemma \ref{lem:RR_min_evict}, we know that, once RR  evicts the complete service from the edge, then it does not fetch any non zero fraction of the service for at least $\left \lceil \frac{M\alpha_{i'}}{\kappa-c\alpha_{i'}-g(\alpha_{i'})\kappa}\right \rceil$ slots. 
	We divide the entire time $T$ into frames of size $f_s=\left \lceil \frac{M}{c}\right \rceil + \left \lceil \frac{M\alpha_{i'}}{\kappa-c\alpha_{i'}-g(\alpha_{i'})\kappa}\right \rceil$ slots. 
	Let us define an event $\mathcal{F}$ as the frame that starts with $\left \lceil \frac{M\alpha_{i'}}{\kappa-c\alpha_{i'}-g(\alpha_{i'})\kappa}\right\rceil$ slots with  $\kappa$ arrivals in each slot and followed by $\left\lceil \frac{M}{c}\right \rceil$ slots of zeros arrivals.  Let the probability of event $\mathcal{F}$ occurring be $\Prob(\mathcal{F})$ (independent of $T$). Conditioned on event $\mathcal{F}$, if optimal static policy is to host non zero fraction of service then we have regret at least  $(\kappa -c\alpha_i +g(\alpha_i)\kappa)\left  \lceil \frac{M\alpha_{i'}}{\kappa-c\alpha_{i'}-g(\alpha_{i'})\kappa}\right \rceil>0$ in a frame where $\alpha_i\ne 0$. If optimal is not to host any fraction of service then we have regret of $c\alpha_{i}\left \lceil\frac{M}{c}\right \rceil>0$ in a frame where $\alpha_{i}\ne 0$. Therefore conditioned on event $\mathcal{F}$ RR always has a finite nonzero regret say $d(>0)$ which is independent of $T$. Therefore,
	\begin{align*}
	\mathcal{R}^{\text{RR}}_S(T)&= \E_{r} \left[\sum_{f=1}^{\lfloor T/f_s \rfloor} d \mathds{1}(\text{event $\mathcal{F}$ occured})\right]\\
	&\ge \left(\frac{T}{f_s}-1\right) \Prob(\mathcal{F})d \\
	&\ge \left(\frac{T}{\frac{M}{\kappa-c\alpha_{i'}-g(\alpha_{i'})\kappa}+ \frac{M}{c} +2}-1\right) \Prob(\mathcal{F})d\\
	&=\Omega(T).
	\end{align*}
	So even in the stochastic case RR observes linear regret.	
\end{proof}

\subsection{Proof of Theorem \ref{thm:reg_stch}(b)}
\begin{lemma}	\label{lem:FTPL_subOpt_prob}
	Probability of hosting a sub-optimal fraction of service $\alpha_j\ne \alpha_{i^\star}$ under the FTPL policy  in time slot $t+1$ is bounded as follows
	\begin{align*}
		\Prob(\rho_{t+1}=\alpha_j)\le \exp\left(-\frac{t^2\Delta_j^2}{16\eta_{t+1}^2}\right) + \exp\left(\frac{-\Delta_j^2t}{2\kappa^2}\right).
	\end{align*} 
\end{lemma}
\begin{proof}
	FTPL hosts $\alpha_j$ fraction of service if $c\alpha_it+g(\alpha_i)R_t+\eta_{t+1}\gamma_i > c\alpha_jt+g(\alpha_j)R_t+\eta_{t+1}\gamma_j$ in time slot $t+1$, for all $i\ne j$.  Let $p_{t,j}=\Prob(\rho_{t+1}=\alpha_{j})$. 
	\begin{align*}
	p_{t,j}\le\Prob(c\alpha_jt&+g(\alpha_j)R_t+\eta_{t+1}\gamma_j \\
	& <  c\alpha_{i^\star}t +g(\alpha_{i^\star})R_t +\eta_{t+1}\gamma_{i^\star})\\
	=\Prob(\eta_{t+1}&(\gamma_{i^\star}-\gamma_j) \\
	&\ge c(\alpha_j-\alpha_{i^\star})t+ (g(\alpha_j)-g(\alpha_{i^\star}))R_t).
	\end{align*}
	There can be two possibilities one is $\alpha_j>\alpha_{i^\star}$ or $\alpha_j<\alpha_{i^\star}$ we bound the probability $p_{t,j}$ by considering each case separately.\\
	\textit{Case 1 ($\alpha_j>\alpha_{i^\star}$):} Since $g(.)$ is a decreasing function $g(\alpha_j)<g(\alpha_{i^\star})$. Let $\mathcal{E}_j$ be the event that  $R_t<\mu t + t\Delta_j/2$.
	\begin{align*}
	p_{t,j}\le& \Prob(\eta_{t+1}(\gamma_{i^\star}-\gamma_j)\ge c(\alpha_j-\alpha_{i^\star})t \\
	&\qquad\qquad + (g(\alpha_j)-g(\alpha_{i^\star}))R_t, \mathcal{E}_j)
	+\Prob(\mathcal{E}_j^c)\\
	\le&\Prob(\eta_{t+1}(\gamma_{i^\star}-\gamma_j)\ge t\Delta_j-(g(\alpha_{i^\star})-g(\alpha_j))t\Delta_j/2)\\
	&+ \Prob(\mathcal{E}_j^c)\\
	\le&\Prob(\eta_{t+1}(\gamma_{i^\star}-\gamma_j)\ge t\Delta_j/2 )+ \Prob(\mathcal{E}_j^c)\\
	\overset{(a)}{\le}&\exp\left(-\frac{t^2\Delta_j^2}{16\eta_{t+1}^2}\right) + \exp\left(\frac{-\Delta_j^2t}{2\kappa^2}\right),
	\end{align*}
	where $(a)$ is obtained by using the fact that complementary CDF of standard Gaussian $Q(x)\le e^{-x^2/2}$ for $x>0$ and Hoeffding's inequality \cite{hoeffding1994probability}.
	
	\textit{Case 2 ($\alpha_j<\alpha_{i^\star}$):} Let $\mathcal{E}_j$ be an event $R_t>\mu t - t\Delta_j/2$.
	\begin{align*}
	p_{t,j}\le& \Prob(\eta_{t+1}(\gamma_{i^\star}-\gamma_j)\ge c(\alpha_j-\alpha_{i^\star})t \\
	&\qquad\qquad+ (g(\alpha_j)-g(\alpha_{i^\star}))R_t, \mathcal{E}_j)+\Prob(\mathcal{E}_j^c)\\
	\le&\Prob(\eta_{t+1}(\gamma_{i^\star}-\gamma_j)\ge t\Delta_j-(g(\alpha_{j})-g(\alpha_{i^\star}))t\Delta_j/2)\\ 
	&+ \Prob(\mathcal{E}_j^c)\\
	\le&\Prob(\eta_{t+1}(\gamma_{i^\star}-\gamma_j)\ge t\Delta_j/2 )+ \Prob(\mathcal{E}_j^c)\\
	\overset{(a)}{\le}&\exp\left(-\frac{t^2\Delta_j^2}{16\eta_{t+1}^2}\right) + \exp\left(\frac{-\Delta_j^2t}{2\kappa^2}\right),
	\end{align*}
	where $(a)$ is obtained by using the fact that $Q(x)\le e^{-x^2/2}$ for $x>0$ and Hoeffding's inequality \cite{hoeffding1994probability}.
	
	Therefore by combining both cases we get 
	\begin{equation}\label{eq:sub_arm_prob}
		\Prob(\rho_{t+1}=\alpha_i)\le \exp\left(-\frac{t^2\Delta_j^2}{16\eta_{t+1}^2}\right) + \exp\left(\frac{-\Delta_j^2t}{2\kappa^2}\right).
	\end{equation} 
\end{proof}

\begin{lemma} \label{lem:FTPL_reg_stch}
	For $M=0$,
	\begin{align*}
	\mathcal{R}_{S}^{\text{FTPL}}(T)\le& \left( \sqrt{2\log K}+\frac{2\sqrt{2h_1}\log K}{\Delta_{\text{min}}} \right) \left( \alpha  + \frac{4\kappa^2}{\alpha}\right)\\
	&+ \frac{16\alpha^2+4\kappa^2}{\Delta_{\text{min}}},
	\end{align*}
	where $h_1= 2\max\{16\alpha^2,2\kappa^2\}$.
\end{lemma}
\begin{proof} Let $t_0=\lceil h_1\frac{\log K}{\Delta_{\text{min}}^2}\rceil$ where $h_1=2\max\{16\alpha^2,2\kappa^2\}$. From \eqref{eq:Regret_stch} we have
	\begin{align*}
	\mathcal{R}_{S}^{\text{FTPL}}(T)\le &\left( \sum_{t=1}^T  c\rho_{t}+g(\rho_t)\mu-\mu_{i^\star}\right)\\
	=&\sum_{t=1}^T \left( \sum_{i} \mu_i \Prob(\rho_{t}=\alpha_{i})\right) -\mu_{i^\star}\\
	=& \sum_{t=1}^T \sum_{i\ne i^\star}\Delta_i \Prob(\rho_{t}=\alpha_{i})\\
	=& \sum_{t=1}^{t_0} \sum_{i\ne i^\star}\Delta_i \Prob(\rho_{t}=\alpha_{i})\\
	&+\sum_{t=t_0+1}^T \sum_{i\ne i^\star}\Delta_i \Prob(\rho_{t}=\alpha_{i})\\
	\overset{(a)}{\le}& \sqrt{2t_0\log K}\left( \alpha+\frac{4\kappa^2}{\alpha}\right)\\
	&+ \sum_{t=t_0+1}^T \sum_{i\ne i^\star}\Delta_i \Prob(\rho_{t}=\alpha_{i}),
	\end{align*}	
	where $(a)$ is obtained by using adversarial regret bound when $M=0$. Now we consider each term in RHS separately and bound them.
	\begin{align*}
		\sqrt{2t_0\log K}\le& \sqrt{2\log K}\left( 1+\frac{\sqrt{h_1\log K}}{\Delta_{\text{min}}}\right).
	\end{align*}
	 By using Lemma \ref{lem:FTPL_subOpt_prob} we have
	\begin{align*}
		\Delta_{j} \Prob(\rho_{t+1}&=\alpha_{j})\\
		\le& \Delta_j\left[  \exp\left(-\frac{t^2\Delta_j^2}{16\eta_{t+1}^2}\right) + \exp\left(\frac{-\Delta_j^2t}{2\kappa^2}\right)\right]  \\
		= & \Delta_j\left[  \exp\left(-\frac{t\Delta_j^2}{16\alpha^2}\right) + \exp\left(\frac{-\Delta_j^2t}{2\kappa^2}\right)\right]\\
		\overset{(b)}{\le}& \Delta_{\text{min}}\left[  \exp\left(-\frac{t\Delta_{\text{min}}^2}{16\alpha^2}\right) + \exp\left(\frac{-\Delta_{\text{min}}^2t}{2\kappa^2}\right)\right].
	\end{align*}
	Here $(b)$ is obtained using the fact that $f(u)=ue^{-u^2}$ is a decreasing function over $u \in [1,\infty)$, and for $t>t_0$, $\frac{t\Delta_{j}^2}{16\alpha^2} > \frac{h_1\log K\Delta_{j}^2}{16\alpha^2\Delta_{\text{min}}^2} \ge 1$ and  $\frac{t\Delta_{j}^2}{2\kappa^2} > \frac{h_1\log K\Delta_{j}^2}{2\kappa^2\Delta_{\text{min}}^2} \ge 1$. Therefore, 
	\begin{align*}
		\mathcal{R}_{S}^{\text{FTPL}}&(T)\\
		 \le&\left(  \sqrt{2\log K}+\frac{2\sqrt{2h_1}\log K}{\Delta_{\text{min}}}\right) \left( \alpha  + \frac{4\kappa^2}{\alpha}\right)\\
		&+\sum_{t=t_0}^{T} K\Delta_{\text{min}}\left[  \exp\left(-\frac{t\Delta_{\text{min}}^2}{16\alpha^2}\right) + \exp\left(\frac{-\Delta_{\text{min}}^2t}{2\kappa^2}\right)\right]\\
		\overset{(c)}{\le}& \left( \sqrt{2\log K}+\frac{2\sqrt{2h_1}\log K}{\Delta_{\text{min}}}\right)  \left( \alpha  + \frac{4\kappa^2}{\alpha}\right)\\
		& + \frac{16\alpha^2+2\kappa^2}{\Delta_{\text{min}}}+2\Delta_{\text{min}}\\
		\le& \left( \sqrt{2\log K}+\frac{2\sqrt{2h_1}\log K}{\Delta_{\text{min}}} \right) \left( \alpha  + \frac{4\kappa^2}{\alpha}\right)\\
		&+ \frac{16\alpha^2+4\kappa^2}{\Delta_{\text{min}}}.
	\end{align*}
	Here $(c)$ follows because for $t>t_0$ we have $K\exp(-t\Delta_{\text{min}}^2/16\alpha^2)<1$, $K\exp(-t\Delta_{\text{min}}^2/2\kappa^2)<1$ and $\sum_{t=1}^{\infty}e^{-at}\le 1/a$ for $a>0$.
\end{proof}

\begin{lemma}	\label{lem:FTPL_swcost_stc}
	Fetch cost under FTPL policy with learning rate $\eta_t=\alpha\sqrt{t-1}$ is bounded as follows	
	$$\E[\mathcal{C}_f^{\text{FTPL}}(T)]\le M \frac{16\alpha^2+2\kappa^2}{\Delta_{\text{min}}^2}.$$
\end{lemma}
\begin{proof}
	The fetch cost is incurred when we fetch extra fraction of service. Similar to the proof of Lemma \ref{lem:FTPL_swcost_adv} we get
	\begin{align*}
	\E[\mathcal{C}^{\text{FTPL}}_{f}(T)]&\le M\sum_{t=1}^T \Prob(\rho_t<\rho_{t+1})\\
	&= M\sum_{t=1}^T \sum_{j>i}\Prob(\rho_t=\alpha_{i}, \rho_{t+1}=\alpha_{j}).
	\end{align*}
	Let $\mathcal{E}_{i,j}^{(t)}$ be the event that $\rho_{t}=\alpha_{i}$, $\rho_{t+1}=\alpha_{j}$. $\mathcal{E}_{i,j}^{(t)}$ occurs if $\Theta_{t,i}+\eta_t \gamma_i< \Theta_{t,j}+\eta_t \gamma_j$, $\Theta_{t+1,i}+\eta_{t+1} \gamma_i> \Theta_{t+1,j}+\eta_{t+1} \gamma_j$.
	\begin{align*}
		\Prob(\mathcal{E}_{i,j}^{(t)}) 	
		&=\Prob\left(\frac{\Theta_{t+1,j}-\Theta_{t+1,i}}{\sqrt{2}\eta_{t+1}} \le \frac{\gamma_i-\gamma_j}{\sqrt{2}} \le \frac{\Theta_{t,j}-\Theta_{t,i}}{\sqrt{2}\eta_{t}}\right). 	
	\end{align*}
	We consider two cases $\mu_i>\mu_j$ and $\mu_i<\mu_j$ and bound the probability.\\
	\textit{Case 1 ($\mu_i>\mu_j$):} Let $\mathcal{E}_1$ be the event that $|R_{t-1}-\mu (t-1)| \le |\Delta_{ij}| (t-1)/2$,
	\begin{align*}
	\Prob(&\mathcal{E}_{i,j}^{(t)})\\
	\le&\Prob\left(\frac{\gamma_i-\gamma_j}{\sqrt{2}}\le \frac{\Theta_{t,j}-\Theta_{t,i}}{\sqrt{2}\eta_{t}}\right)\\
	\le&\Prob\left(\frac{\gamma_i-\gamma_j}{\sqrt{2}}\le \frac{\Theta_{t,j}-\Theta_{t,i}}{\sqrt{2}\eta_{t}},\mathcal{E}_1\right)+\Prob(\mathcal{E}_1^c)\\
	=&\Prob\left(\frac{\gamma_i-\gamma_j}{\sqrt{2}}\le \frac{c(\alpha_{j}-\alpha_i)(t-1)}{\sqrt{2}\eta_{t}}\right.\\
	&\hspace{2.7cm} \left. -\frac{(g(\alpha_i)-g(\alpha_j))R_{t-1}}{\sqrt{2}\eta_{t}} ,\mathcal{E}_1\right)+\Prob(\mathcal{E}_1^c)\\
	\le&\Prob\left(\frac{\gamma_i-\gamma_j}{\sqrt{2}}\le \frac{-|\Delta_{ij}|(t-1)}{2\sqrt{2}\eta_{t}}\right)+\Prob(\mathcal{E}_1^c)\\
	\overset{(a)}{\le}& \exp\left(-\frac{(t-1)^2\Delta_{ij}^2}{16\eta_{t}^2}\right)+ \exp\left(\frac{-\Delta_{ij}^2(t-1)}{2\kappa^2}\right)\\
	\le& \exp \left(-\frac{(t-1)^2 \Delta_{\text{min}}^2} {16\eta_{t}^2}\right)+ \exp\left(\frac{-\Delta_{\text{min}}^2(t-1)}{2\kappa^2}\right)\\
	\le& \exp \left(-\frac{(t-1) \Delta_{\text{min}}^2} {16\alpha^2}\right)+ \exp\left(\frac{-\Delta_{\text{min}}^2(t-1)}{2\kappa^2}\right),
	\end{align*}
	where $(a)$ is obtained by using the fact that $Q(x)\le e^{-x^2/2}$ for $x>0$ and Hoeffding's inequality \cite{hoeffding1994probability}.
	
	\textit{Case 2 ($\mu_i<\mu_j$):} Let $\mathcal{E}_2$ be the event that $|R_{t}-\mu t| \le |\Delta_{ij}| t/2$,
	\begin{align*}
		\Prob&(\mathcal{E}_{i,j}^{(t)})\\
		\le&\Prob\left(\frac{\Theta_{t+1,j}-\Theta_{t+1,i}}{\sqrt{2}\eta_{t+1}} \le \frac{\gamma_i-\gamma_j}{\sqrt{2}}\right)\\
		\le&\Prob\left(\frac{c(\alpha_{j}-\alpha_i)t-(g(\alpha_i)-g(\alpha_j))R_{t}}{\sqrt{2}\eta_{t}} \le \frac{\gamma_i-\gamma_j}{\sqrt{2}}\right)\\ 	
		\le&\Prob\left(\frac{\gamma_i-\gamma_j}{\sqrt{2}}\ge \frac{c(\alpha_{j}-\alpha_i)t-(g(\alpha_i)-g(\alpha_j))R_{t}}{\sqrt{2}\eta_{t}} ,\mathcal{E}_2\right)\\
		&+\Prob(\mathcal{E}_2^c)\\
		\le&\Prob\left(\frac{\gamma_i-\gamma_j}{\sqrt{2}}\ge \frac{|\Delta_{ij}|t}{2\sqrt{2}\eta_{t}}\right)+\Prob(\mathcal{E}_2^c)\\
		\overset{(b)}{\le}& \exp\left(-\frac{t^2\Delta_{ij}^2}{16\eta_{t}^2}\right)+ \exp\left(\frac{-\Delta_{ij}^2t}{2\kappa^2}\right)\\
		\le& \exp \left(-\frac{t^2 \Delta_{\text{min}}^2} {16\eta_{t}^2}\right)+ \exp\left(\frac{-\Delta_{\text{min}}^2t}{2\kappa^2}\right)\\
		\le& \exp \left(-\frac{t \Delta_{\text{min}}^2} {16\alpha^2}\right)+ \exp\left(\frac{-\Delta_{\text{min}}^2t}{2\kappa^2}\right)\\
		\le& \exp \left(-\frac{(t-1) \Delta_{\text{min}}^2} {16\alpha^2}\right)+ \exp\left(\frac{-\Delta_{\text{min}}^2(t-1)}{2\kappa^2}\right),
	\end{align*}
	where $(b)$ is obtained by using the fact that $Q(x)\le e^{-x^2/2}$ for $x>0$ and Hoeffding's inequality \cite{hoeffding1994probability}.
	
	From combining both cases we get
	\begin{align}\label{eq:swFTPL_inq}
	\Prob(\mathcal{E}_{i,j}^{(t)})\le \exp \left(-\frac{(t-1) \Delta_{\text{min}}^2} {16\alpha^2}\right)+ \exp\left(\frac{-\Delta_{\text{min}}^2(t-1)}{2\kappa^2}\right) .
	\end{align}
	Therefore, 
	\begin{align*}
	\E[\mathcal{C}_F^{\text{FTPL}}]\le& M \sum_{t=1}^{T}\sum_{j>i}\Prob(\mathcal{E}_{i,j}^{(t)})\\
	\le &MK^2\sum_{t=1}^{T-1} \exp \left(-\frac{t \Delta_{\text{min}}^2}{16\alpha^2}\right)+ \exp\left(\frac{-\Delta_{\text{min}}^2t}{2\kappa^2}\right)\\
	&+MK^2\\
	\le& MK^2\frac{16\alpha^2+3\kappa^2}{\Delta_{\text{min}}^2}.
	\end{align*}
\end{proof}

\begin{proof}[Proof of Theorem \ref{thm:reg_stch}(b)]
	Note that FTPL policy does not consider fetch cost $M$ while taking the decisions. By using  Lemma \ref{lem:FTPL_reg_stch}, \ref{lem:FTPL_swcost_stc} we get the result stated.
\end{proof}

\subsection{Proof of Theorem \ref{thm:reg_stch}(c)}
\begin{lemma}\label{lem:T_s,mu_hat} 
	Under the W-FTPL policy, $T_s>t'$ if and only if there exist $t<t'$ such that the following condition holds:
	\begin{align*}
	(g(\alpha_{\tilde{j}_t})-g(\alpha_{\tilde{i}_t}))\hat{\mu}_t &> -c(\alpha_{\tilde{j}_t}-\alpha_{\tilde{i}_t})+\kappa\sqrt{\frac{\beta (\log M)^{1+\delta}}{t}},\\ \text{ or}\\
	(g(\alpha_{\tilde{j}_t})-g(\alpha_{\tilde{i}_t}))\hat{\mu}_t &< -c(\alpha_{\tilde{j}_t}-\alpha_{\tilde{i}_t})-\kappa\sqrt{\frac{\beta (\log M)^{1+\delta}}{t}},
	\end{align*}
	where $(\tilde{i}_t,\tilde{j}_t)=\argmax_{i\ne j}\Theta_{t+1,i}-\Theta_{t+1,j}$, $\hat{\mu}_t=R_t/t$.
\end{lemma}
\begin{proof}
	Let $(\tilde{i}_t,\tilde{j}_t)=\argmax_{i\ne j}\Theta_{t+1,i}-\Theta_{t+1,j}$, $\hat{\mu}_t=R_t/t$. $T_s>t'$ if and only if there exist $t<t'$ such that
	\begin{align*}
		\frac{|\max_{i\ne j}(\Theta_{t+1,i}-\Theta_{t+1,j})|}{t}&< \kappa\sqrt{\frac{\beta (\log M)^{1+\delta}}{t}}\\
		\iff \frac{|\Theta_{t+1,\tilde{j}_t}-\Theta_{t+1,\tilde{i}_t|}}{t} &> \kappa\sqrt{\frac{\beta (\log M)^{1+\delta}}{t}}\\
		&\Updownarrow\\
		 |c(\alpha_{\tilde{j}_t}-\alpha_{\tilde{i}_t}) + (g(\alpha_{\tilde{j}_t})-g(\alpha_{\tilde{i}_t}))\hat{\mu}_t| &> \kappa\sqrt{\frac{\beta (\log M)^{1+\delta}}{t}}\\
		&\Updownarrow\\
		(g(\alpha_{\tilde{j}_t})-g(\alpha_{\tilde{i}_t}))\hat{\mu}_t > -c(\alpha_{\tilde{j}_t}-\alpha_{\tilde{i}_t})&+\kappa\sqrt{\frac{\beta (\log M)^{1+\delta}}{t}}\\
		&\text{or } \\
		(g(\alpha_{\tilde{j}_t})-g(\alpha_{\tilde{i}_t}))\hat{\mu}_t < -c(\alpha_{\tilde{j}_t}-\alpha_{\tilde{i}_t})&-\kappa\sqrt{\frac{\beta (\log M)^{1+\delta}}{t}}.
	\end{align*}
This completes the proof. 
\end{proof}

\begin{lemma}\label{lem:Ts>t0}
	Under the W-FTPL policy, $T_s>\frac{(\sqrt{\beta}-1)^2 (\log M)^{1+\delta}}{\Delta_{\text{max}}^2}$ with probability at least $1-\frac{(\sqrt{\beta}-1)^2 (\log M)^{1+\delta}}{M^{2(\log M)^\delta}\Delta_{\text{max}}^2}$.
\end{lemma}
	
\begin{proof}
	Let $\mathcal{E}_t$ be the event that $$|\mu-\hat{\mu}_t|\le \kappa\sqrt{\frac{(\log M)^{1+\delta}}{t}},$$
	 where $\hat{\mu}_t=R_t/t$. By using Hoeffding's inequality \cite{hoeffding1994probability} we get $$\Prob(\mathcal{E}_t)\ge 1-\frac{1}{M^{2(\log M)^\delta}}.$$  Let $\mathcal{E}=\cap_{t=1}^{T_0}\mathcal{E}_t$, where $T_0=\frac{(\sqrt{\beta}-1)^2 \kappa^2(\log M)^{1+\delta}}{\Delta_{\text{max}}^2}$. Using the union bound, we get: 
	 $$\Prob(\mathcal{E})\ge 1-\frac{T_0}{M^{2(\log M)^\delta}}.$$ 
	 Let $\epsilon_t=\kappa\sqrt{\frac{\beta (\log M)^{1+\delta}}{t}}.$ We prove the lemma by contradiction. Let us consider $$T_s<\frac{(\sqrt{\beta}-1)^2 \kappa^2(\log M)^{1+\delta}}{\Delta_{\text{max}}^2},$$ 
	 and $\mathcal{E}$ holds, then by Lemma \ref{lem:T_s,mu_hat}, $\exists t<\frac{(\sqrt{\beta}-1)^2 \kappa^2(\log M)^{1+\delta}}{\Delta_{\text{max}}^2}$ such that \begin{align*}
	 	 (g(\alpha_{\tilde{j}_t})-g(\alpha_{\tilde{i}_t}))\hat{\mu}_t &> -c(\alpha_{\tilde{j}_t}-\alpha_{\tilde{i}_t})+\kappa\sqrt{\frac{\beta (\log M)^{1+\delta}}{t}},\\ &\text{ or}\\ (g(\alpha_{\tilde{j}_t})-g(\alpha_{\tilde{i}_t}))\hat{\mu}_t &< -c(\alpha_{\tilde{j}_t}-\alpha_{\tilde{i}_t})-\kappa\sqrt{\frac{\beta (\log M)^{1+\delta}}{t}},
	  \end{align*}
	  holds. There are two possibilities $\alpha_{\tilde{j}_t}<\alpha_{\tilde{i}_t}$, $\alpha_{\tilde{j}_t}>\alpha_{\tilde{i}_t}$  we consider them separately to prove the final result.\\
	 \textit{Case 1 ($\alpha_{\tilde{j}_t}<\alpha_{\tilde{i}_t}$):} It follows that
	\begin{align*}
		&(g(\alpha_{\tilde{j}_t})-g(\alpha_{\tilde{i}_t}))\hat{\mu}_t > -c(\alpha_{\tilde{j}_t}-\alpha_{\tilde{i}_t})+\epsilon_t\\
		&\implies \hat{\mu}_t-\mu >\frac{c(\alpha_{\tilde{i}_t}-\alpha_{\tilde{j}_t})+\epsilon_t} {g(\alpha_{\tilde{j}_t})-g(\alpha_{\tilde{i}_t})}-\mu\\
		&\implies \kappa\sqrt{\frac{(\log M)^{1+\delta}}{t}} > \frac{-\Delta_{\text{max}}+\epsilon_t} {g(\alpha_{\tilde{j}_t})-g(\alpha_{\tilde{i}_t})}\\
		&\implies (\sqrt{\beta}-g(\alpha_{\tilde{j}_t})+g(\alpha_{\tilde{i}_t}))\kappa\sqrt{\frac{(\log M)^{1+\delta}}{t}}< \Delta_{\text{max}}\\
		&\implies (\sqrt{\beta}-1))\kappa\sqrt{\frac{(\log M)^{1+\delta}}{t}}< \Delta_{\text{max}}\\
		&\implies t>\frac{(\sqrt{\beta}-1)^2 \kappa^2 (\log M)^{1+\delta}}{\Delta_{\text{max}}^2},
	\end{align*}
	which is a contradiction. Alternatively,
	\begin{align*}
		&(g(\alpha_{\tilde{j}_t})-g(\alpha_{\tilde{i}_t}))\hat{\mu}_t < -c(\alpha_{\tilde{j}_t}-\alpha_{\tilde{i}_t})-\epsilon_t\\
		&\implies \hat{\mu}_t-\mu <\frac{c(\alpha_{\tilde{i}_t}-\alpha_{\tilde{j}_t})-\epsilon_t} {g(\alpha_{\tilde{j}_t})-g(\alpha_{\tilde{i}_t})}-\mu\\
		&\implies -\kappa\sqrt{\frac{(\log M)^{1+\delta}}{t}} < \frac{\Delta_{\text{max}}-\epsilon_t} {g(\alpha_{\tilde{j}_t})-g(\alpha_{\tilde{i}_t})}\\
		&\implies (\sqrt{\beta}-g(\alpha_{\tilde{j}_t})+g(\alpha_{\tilde{i}_t}))\kappa\sqrt{\frac{(\log M)^{1+\delta}}{t}}< \Delta_{\text{max}}\\
		&\implies (\sqrt{\beta}-1))\kappa\sqrt{\frac{(\log M)^{1+\delta}}{t}}< \Delta_{\text{max}}\\
		&\implies t>\frac{(\sqrt{\beta}-1)^2 \kappa^2(\log M)^{1+\delta}}{\Delta_{\text{max}}^2},
	\end{align*}
		which is also a contradiction. \\
	\textit{Case 2 }($\alpha_{\tilde{j}_t}>\alpha_{\tilde{i}_t}$): It follows that:
	\begin{align*}
		&	(g(\alpha_{\tilde{j}_t})-g(\alpha_{\tilde{i}_t}))\hat{\mu}_t > -c(\alpha_{\tilde{j}_t}-\alpha_{\tilde{i}_t})+\epsilon_t\\
		&\implies \hat{\mu}_t-\mu <\frac{c(\alpha_{\tilde{j}_t}-\alpha_{\tilde{i}_t})-\epsilon_t} {g(\alpha_{\tilde{i}_t})-g(\alpha_{\tilde{j}_t})}-\mu\\
		&\implies -\kappa\sqrt{\frac{(\log M)^{1+\delta}}{t}} < \frac{\Delta_{\text{max}}-\epsilon_t} {g(\alpha_{\tilde{i}_t})-g(\alpha_{\tilde{j}_t})}\\
		&\implies t>\frac{(\sqrt{\beta}-1)^2 \kappa^2(\log M)^{1+\delta}}{\Delta_{\text{max}}^2},
	\end{align*}
	which is a contradiction. Alternatively,
	\begin{align*}
	&(g(\alpha_{\tilde{j}_t})-g(\alpha_{\tilde{i}_t}))\hat{\mu}_t < -c(\alpha_{\tilde{j}_t}-\alpha_{\tilde{i}_t})-\epsilon_t\\
	&\implies \hat{\mu}_t-\mu >\frac{c(\alpha_{\tilde{j}_t}-\alpha_{\tilde{i}_t})+\epsilon_t} {g(\alpha_{\tilde{i}_t})-g(\alpha_{\tilde{j}_t})}-\mu\\
	&\implies \kappa\sqrt{\frac{(\log M)^{1+\delta}}{t}} > \frac{-\Delta_{\text{max}}+\epsilon_t}{g(\alpha_{\tilde{i}_t})-g(\alpha_{\tilde{j}_t})}\\
	&\implies t>\frac{(\sqrt{\beta}-1)^2 \kappa^2(\log M)^{1+\delta}}{\Delta_{\text{max}}^2},
	\end{align*}
	which is also a contradiction. 
	
	Therefore $T_s>\frac{(\sqrt{\beta}-1)^2 \kappa^2(\log M)^{1+\delta}}{\Delta_{\text{max}}^2}$, if $\mathcal{E}$ occurs, which happens with probability at least $1-\frac{(\sqrt{\beta}-1)^2 \kappa^2(\log M)^{1+\delta}}{M^{2(\log M)^\delta}\Delta_{\text{max}}^2}$.
\end{proof}

\begin{lemma}
	Under the W-FTPL policy, $\Prob(T_s>t)\le \exp(-\frac{\Delta_{\text{min}}^2 t}{2\kappa^2})$ for $t>\frac{4\beta \kappa^2(\log M)^{1+\delta}}{\Delta_{\text{min}}^2}$. \label{lem:Ts>T_0}
\end{lemma}
\begin{proof} For $t>\frac{4\beta \kappa^2 (\log M)^{1+\delta}}{\Delta_{\text{min}}^2}$, we have $\frac{\Delta_{\text{min}}}{2}>\kappa\sqrt{\frac{\beta (\log M)^{1+\delta}}{t}}$. For $T_s>t$, we have
		\begin{align*} 
		(g(\alpha_{\tilde{j}_t})-g(\alpha_{\tilde{i}_t}))\hat{\mu}_t &< -c(\alpha_{\tilde{j}_t}-\alpha_{\tilde{i}_t})+\kappa\sqrt{\frac{\beta (\log M)^{1+\delta}}{t}},\\
		 \text{ and}\\
		(g(\alpha_{\tilde{j}_t})-g(\alpha_{\tilde{i}_t}))\hat{\mu}_t &> -c(\alpha_{\tilde{j}_t}-\alpha_{\tilde{i}_t})-\kappa\sqrt{\frac{\beta (\log M)^{1+\delta}}{t}}.
	\end{align*}
	Let us consider the following cases \\
	\textit{Case 1} ($\alpha_{\tilde{j}_t}<\alpha_{\tilde{i}_t}$, $\Delta_{\tilde{i}_t\tilde{j}_t}>0$):\\
	\begin{align*}
		\Prob(T_s>t)&\le\Prob\left(\hat{\mu}_t > \frac{-c(\alpha_{\tilde{j}_t}-\alpha_{\tilde{i}_t})- \kappa\sqrt{\frac{\beta (\log M)^{1+\delta}}{t}}}{(g(\alpha_{\tilde{j}_t})-g(\alpha_{\tilde{i}_t}))}\right)\\
		&\le \Prob\left( \hat{\mu}_t-\mu>\frac{\Delta_{\tilde{i}_t\tilde{j}_t}-\kappa\sqrt{\frac{\beta (\log M)^{1+\delta}}{t}}}{g(\alpha_{\tilde{j}_t})-g(\alpha_{\tilde{i}_t})}\right) \\
		&\le \Prob\left( \hat{\mu}_t-\mu>\frac{\Delta_{\text{min}}-\Delta_{\text{min}}/2} {g(\alpha_{\tilde{j}_t})-g(\alpha_{\tilde{i}_t})}\right) \\
		&\le \Prob\left( \hat{\mu}_t-\mu>\frac{\Delta_{\text{min}}}{2}\right) \\
		&\overset{(a)}{\le} \exp\left(\frac{-\Delta_{\text{min}}^2 t}{2\kappa^2}\right),
	\end{align*}
	where $(a)$ is obtained by using Hoeffding's inequality \cite{hoeffding1994probability}.\\
	\textit{Case 2} ($\alpha_{\tilde{j}_t}<\alpha_{\tilde{i}_t}$, $\Delta_{\tilde{i}_t\tilde{j}_t}<0$):
		\begin{align*}
		\Prob(T_s>t)&\le\Prob\left(\hat{\mu}_t < \frac{-c(\alpha_{\tilde{j}_t}-\alpha_{\tilde{i}_t})+ \kappa\sqrt{\frac{\beta (\log M)^{1+\delta}}{t}}}{(g(\alpha_{\tilde{j}_t})-g(\alpha_{\tilde{i}_t}))}\right)\\
		&\le \Prob\left( \hat{\mu}_t-\mu<\frac{\Delta_{\tilde{i}_t\tilde{j}_t}+\kappa\sqrt{\frac{\beta (\log M)^{1+\delta}}{t}}}{g(\alpha_{\tilde{j}_t})-g(\alpha_{\tilde{i}_t})}\right) \\
		&\le \Prob\left( \hat{\mu}_t-\mu<\frac{-\Delta_{\text{min}}+\Delta_{\text{min}}/2} {g(\alpha_{\tilde{j}_t})-g(\alpha_{\tilde{i}_t})}\right) \\
		&\le \Prob\left( \hat{\mu}_t-\mu<\frac{-\Delta_{\text{min}}}{2}\right) \\
		&\overset{(a)}{\le} \exp\left(\frac{-\Delta_{\text{min}}^2 t}{2\kappa^2}\right).
	\end{align*}	
	 where $(a)$ is obtained by using Hoeffding's inequality \cite{hoeffding1994probability}.\\
	\textit{Case 3} ($\alpha_{\tilde{j}_t}>\alpha_{\tilde{i}_t}$, $\Delta_{\tilde{i}_t\tilde{j}_t}>0$):\\
	\begin{align*}
	\Prob(T_s>t)&\le\Prob\left(\hat{\mu}_t < \frac{-c(\alpha_{\tilde{j}_t}-\alpha_{\tilde{i}_t})- \kappa\sqrt{\frac{\beta (\log M)^{1+\delta}}{t}}}{(g(\alpha_{\tilde{j}_t})-g(\alpha_{\tilde{i}_t}))}\right)\\
	&\le \Prob\left( \hat{\mu}_t-\mu<\frac{-\Delta_{\tilde{i}_t\tilde{j}_t}+\kappa\sqrt{\frac{\beta (\log M)^{1+\delta}}{t}}}{g(\alpha_{\tilde{i}_t})-g(\alpha_{\tilde{j}_t})}\right) \\
	&\le \Prob\left( \hat{\mu}_t-\mu<\frac{-\Delta_{\text{min}}+\kappa\sqrt{\frac{\beta (\log M)^{1+\delta}}{t}}}{g(\alpha_{\tilde{i}_t})-g(\alpha_{\tilde{j}_t})}\right) \\
	&\le \Prob\left( \hat{\mu}_t-\mu<\frac{-\Delta_{\text{min}}/2} {g(\alpha_{\tilde{i}_t})-g(\alpha_{\tilde{j}_t})}\right) \\
	&\le \Prob\left( \hat{\mu}_t-\mu<-\Delta_{\text{min}}/2\right) \\
	&\overset{(a)}{\le} \exp\left(\frac{-\Delta_{\text{min}}^2 t}{2\kappa^2}\right).
	\end{align*}
	 where $(a)$ is obtained by using Hoeffding's inequality \cite{hoeffding1994probability}.\\
	\textit{Case 4} ($\alpha_{\tilde{j}_t}>\alpha_{\tilde{i}_t}$, $\Delta_{\tilde{i}_t\tilde{j}_t}<0$):\\
	\begin{align*}
	\Prob(T_s>t)&\le\Prob\left(\hat{\mu}_t > \frac{-c(\alpha_{\tilde{j}_t}-\alpha_{\tilde{i}_t})+ \kappa\sqrt{\frac{\beta (\log M)^{1+\delta}}{t}}}{(g(\alpha_{\tilde{j}_t})-g(\alpha_{\tilde{i}_t}))}\right)\\
	&\le \Prob\left( \hat{\mu}_t-\mu>\frac{-\Delta_{\tilde{i}_t\tilde{j}_t}-\kappa\sqrt{\frac{\beta (\log M)^{1+\delta}}{t}}}{g(\alpha_{\tilde{i}_t})-g(\alpha_{\tilde{j}_t})}\right) \\
	&\le \Prob\left( \hat{\mu}_t-\mu>\frac{\Delta_{\text{min}}-\Delta_{\text{min}}/2} {g(\alpha_{\tilde{i}_t})-g(\alpha_{\tilde{j}_t})}\right) \\
	&\le \Prob\left( \hat{\mu}_t-\mu>\frac{\Delta_{\text{min}}}{2}\right) \\
	&\overset{(a)}{\le} \exp\left(\frac{-\Delta_{\text{min}}^2 t}{2\kappa^2}\right), 
	\end{align*}
	 where $(a)$ is obtained by using Hoeffding's inequality \cite{hoeffding1994probability}.
	The result follows by combining these cases. 
\end{proof}

\begin{lemma}
	Under the W-FTPL policy 
	$$\E[T_s]\le 1+\frac{4\beta \kappa^2(\log M)^{1+\delta}}{\Delta_{\text{min}}^2} + \frac{1}{M^{2\beta}}\frac{2\kappa^2}{\Delta_{\text{min}}^2}.$$
	\label{lem:E[T_s]}
\end{lemma}
\begin{proof}Let $T_1=\lceil \frac{4\beta \kappa^2 (\log M)^{1+\delta}}{\Delta_{\text{min}}^2}\rceil$. It follows that
	\begin{align*}
		\E[T_s]&=\sum_{t=1}^{T}\Prob(T_s>t)\\
		&=\sum_{t=1}^{T_1}\Prob(T_s>t)+\sum_{t=T_1+1}^{T}\Prob(T_s>t)\\
		&\overset{(a)}{\le} T_1+\sum_{t=T_1}^{T}\exp\left(-\frac{\Delta_{\text{min}}^2 t}{2\kappa^2}\right)\\
		&\le 1+\frac{4\beta \kappa^2(\log M)^{1+\delta}}{\Delta_{\text{min}}^2} + \frac{1}{M^{2\beta(\log M)^\delta}}\frac{2\kappa^2}{\Delta_{\text{min}}^2},
	\end{align*}
	 where $(a)$ is obtained by using Lemma \ref{lem:Ts>T_0}.
\end{proof}

\begin{proof}[Proof of Theorem \ref{thm:reg_stch}(c)]
	Under W-FTPL,
	\begin{align*}
		\E_r[\mathcal{C}^\text{W-FTPL}(T,r)]=& \E_r[T_s]+ \E_r\left[ \sum_{t=T_s+1}^{T}(c\rho_t+ g(\rho_t)r_t)\right] \\
		&+M\E_r\left[ \sum_{t=T_s+1}^T \Prob(\rho_t>\rho_{t-1})\right] .	
	\end{align*}
	By the using definition of regret we get,
	\begin{align*}
		\mathcal{R}^{\text{W-FTPL}}_S(T)\le& \E[T_s]+\sum_{t=1}^{T}(c\rho_t+\mu g(\rho_t))-\mu_{i^\star}\\
		&+M\E_r\left[ \sum_{t=T_s+1}^T \Prob(\rho_t>\rho_{t-1})\right] .
	\end{align*}
	 By Lemma \ref{lem:FTPL_reg_stch}, we get
	\begin{align*}
		\mathcal{R}^{\text{W-FTPL}}_S(T)\le& \E[T_s]+ \frac{16\alpha^2+4\kappa^2}{\Delta_{\text{min}}}\\
		&+\left( \sqrt{2\log K}+\frac{2\sqrt{2h_1}\log K}{\Delta_{\text{min}}} \right) \left( \alpha  + \frac{4\kappa^2}{\alpha}\right)\\
		&+M\E_r\left[ \sum_{t=T_s+1}^T \Prob(\rho_t>\rho_{t-1})\right] .		
	\end{align*}
	Let $\E_r[\mathcal{C}_f]=M\E_r\left[ \sum_{t=T_s+1}^T \Prob(\rho_t>\rho_{t-1})\right] $, then 
	\begin{align*}
		\E_r[\mathcal{C}_f]=&\E_r[\mathcal{C}_f|T_s\le T_0]\Prob(T_s\le T_0)\\
		&+\E_r[\mathcal{C}_f|T_s> T_0]\Prob(T_s>T_0)\\
		\le &\E[\mathcal{C}_f|T_s=1]\Prob(T_s\le T_0)+\E[\mathcal{C}_f|T_s= \lceil T_0 \rceil]\\
		\overset{(a)}{\le} &M \Prob(T_s\le T_0)\E[\mathcal{C}_F^{\text{FTPL}}]\\
		&+M \sum_{t=\lceil T_0\rceil}^T \left( \exp\left( -\frac{\Delta_{\text{min}}^2t}{16\alpha^2}\right)+e^{-\Delta_{\text{min}}^2t/2\kappa^2}\right)\\
		\le & M \Prob(T_s\le T_0) \frac{16\alpha^2+3\kappa^2}{\Delta_{\text{min}}^2}\\
		&+ \frac{16\alpha^2}{\Delta_{\text{min}}^2}\exp\left( -\frac{\Delta_{\text{min}}^2T_0}{16\alpha^2}\right)\\
		&+\frac{2\kappa^2}{\Delta_{\text{min}}^2}\exp\left(-\frac{\Delta_{\text{min}}^2T_0}{2\kappa^2}\right).
	\end{align*}
	Here, $(a)$ is obtained by using inequality \eqref{eq:swFTPL_inq}.
	By considering $T_0=\frac{(\sqrt{\beta}-1)^2 \kappa^2(\log M)^{1+\delta}}{\Delta_{\text{max}}^2}$ and using Lemma \ref{lem:Ts>t0}, \ref{lem:E[T_s]} we get,
	\begin{align*}
	\mathcal{R}^{\text{W-FTPL}}_S&(T) \\
	\le& 1+\frac{4\beta \kappa^2(\log M)^{1+\delta}}{\Delta_{\text{min}}^2} + \frac{1}{M^{2\beta(\log M)^\delta}}\frac{2\kappa^2	}{\Delta_{\text{min}}^2}\\
	&+\frac{(16\alpha^2+4\kappa^2)}{\Delta_{\text{min}}}\\
	&+\left( \sqrt{2\log K}+\frac{2\sqrt{2h_1}\log K}{\Delta_{\text{min}}} \right) \left( \alpha  + \frac{4\kappa^2}{\alpha}\right)\\
	&+M\frac{(\sqrt{\beta}-1)^2\kappa^2 (\log M)^{1+\delta}}{\Delta_{\text{max}}^2 M^{2(\log M)^\delta}}\frac{16\alpha^2+3\kappa^2}{\Delta_{\text{min}}^2}\\
	&+\left( \frac{16\alpha^2 M}{M^{\frac{(\sqrt{\beta}-1)^2\kappa^2\Delta_{\text{min}}^2(\log M)^\delta}{16\alpha^2\Delta_{\text{max}}^2}}}+ \frac{2M\kappa^2}{M^{\frac{(\sqrt{\beta}-1)^2\Delta_{\text{min}}^2(\log M)^\delta}{2\Delta_{\text{max}}^2}}}	\right)\\ &\times \frac{1}{\Delta_{\text{min}}^2}.
	\end{align*}

	For large values of $M$ we have
	\begin{align*}
		\mathcal{R}^{\text{W-FTPL}}_S&(T) \\
		\le& 1+\frac{4\beta \kappa^2(\log M)^{1+\delta}}{\Delta_{\text{min}}^2} +\frac{2\kappa^2	}{\Delta_{\text{min}}^2}\\
		&+\frac{(16\alpha^2+4\kappa^2)}{\Delta_{\text{min}}}\\
		&+\left( \sqrt{2\log K}+\frac{2\sqrt{2h_1}\log K}{\Delta_{\text{min}}} \right) \left( \alpha  + \frac{4\kappa^2}{\alpha}\right)\\
		&+\frac{(\sqrt{\beta}-1)^2\kappa^2 (\log M)^{1+\delta}}{\Delta_{\text{max}}^2 }\frac{16\alpha^2+3\kappa^2}{\Delta_{\text{min}}^2}\\
		&+\frac{16\alpha^2 + 2\kappa^2}{\Delta_{\text{min}}^2}\\
		\le&1+\beta \kappa^2(\log M)^{1+\delta}\left(\frac{4}{\Delta_{\text{min}^2}}+\frac{16\alpha^2+3\kappa^2}{\Delta_{\text{min}^2}\Delta_{\text{max}}^2}\right)\\
		&+(16\alpha^2+4\kappa^2)\left(\frac{1}{\Delta_{\text{min}}}+\frac{1}{\Delta_{\text{min}^2}}\right)\\
		&+\left( \sqrt{2\log K}+\frac{2\sqrt{2h_1}\log K}{\Delta_{\text{min}}} \right) \left( \alpha  + \frac{4\kappa^2}{\alpha}\right).
	\end{align*}
\end{proof}

\subsection{Proof of Theorem \ref{thm:comp_adv}(a)} 

\begin{lemma}
	The cost of the offline optimal policy is lower bounded as follows 
	$$\mathcal{C}^\text{OPT-OFF}(T,r)\ge \frac{R_T}{\kappa^2}(\min_i{c\alpha_{i}+g(\alpha_{i})\kappa}).$$
	\label{lem:opt_offline_cost_LB}
\end{lemma}
\begin{proof}
	The cost incurred by the optimal offline policy with $M>0$ is lower bounded by the cost incurred by the optimal offline policy when $M=0$. For $M>0$, there is an additional fetching cost. Therefore, 
	\begin{align*}
		\mathcal{C}^\text{OPT-OFF}(T,r)&\ge \sum_{t=1}^T \min_i (c\alpha_{i}+g(\alpha_{i})r_t)\\
		&\ge \sum_{t=1}^T \min_i (c\alpha_{i}/\kappa+g(\alpha_{i})r_t)\\
		&\ge \sum_{t=1}^T \min_i (c\alpha_{i}/\kappa+g(\alpha_{i}))\mathds{1}_{\{r_t\ge 1\}}\\
		&\ge \sum_{t=1}^T \min_i (c\alpha_{i}/\kappa+g(\alpha_{i})) r_t/\kappa\\
		&\ge \frac{R_T}{\kappa^2}(\min_i{c\alpha_{i}+g(\alpha_{i})\kappa}).
	\end{align*}
\end{proof}

\begin{lemma}
	The expected cost of FTPL policy for any request sequence $r$ for $\eta_{t}=\alpha\sqrt{t-1}$, $\alpha>0$, can be upper bounded as follows 
	\begin{align*}
		\E[\mathcal{C}^{\text{FTPL}}(T,r)]\le \left((3+2M/c)\max_{\alpha_{i}\ne 0}\frac{1-g(\alpha_{i})}{\alpha_{i}}\right) R_T\\+(c+M)\sum_{i=2}^K\frac{16\alpha^2}{c^2\alpha_{i}^2}.
	\end{align*}
	\label{lem:FTPL_comp_adv}
\end{lemma}
\begin{proof}
	FTPL does not host any fraction the service for time $t+1$ if $c\alpha_{i}t+g(\alpha_{i})R_t+\eta_{t+1}\gamma_i> R_t+\eta_{t+1} \gamma_1$, for all $i\ne 1$. The probability of hosting any non zero fraction of service ($p_{h,t}$) in time slot $t+1$ is given as 
	\begin{align*}
	p_{h,t}&\le\sum_{i=2}^K\Prob(c\alpha_{i}t+g(\alpha_{i})R_t+\eta_{t+1}\gamma_i< R_t+\eta_{t+1} \gamma_1)\\
	&=\sum_{i=2}^K\Prob\left( \frac{\gamma_1-\gamma_i}{\sqrt{2}}> \frac{c\alpha_{i}t+g(\alpha_{i})R_t-R_t}{\sqrt{2}\eta_{t+1}}\right)\\
	&=\sum_{i=2}^K\Prob\left( \frac{\gamma_1-\gamma_i}{\sqrt{2}}> \frac{c\alpha_{i}t-(1-g(\alpha_{i}))R_t}{\sqrt{2}\eta_{t+1}}\right).
	\end{align*}
	 Let $T'=\frac{2R_T}{c}\max_{\alpha_{i}\ne 0}\frac{1-g(\alpha_{i})}{\alpha_{i}}$, if  $t\ge T'$ then $R_t\le R_T\le \frac{ct}{2}\min_{\alpha_{i}\ne0}\frac{\alpha_{i}}{1-g(\alpha_{i})}$. Therefore  for $t\ge T'$, $R_t\le \frac{c\alpha_{i}t}{2(1-g(\alpha_{i}))}$ for all $\alpha_{i}\ne 0$ and we get
	\begin{align*}
	p_{h,t}&\le \sum_{i=2}^K\Prob\left( \frac{\gamma_1-\gamma_i}{\sqrt{2}}> \frac{c\alpha_{i}t-c\alpha_{i}t/2}{\sqrt{2}\eta_{t+1}}\right)\\
	&\le \sum_{i=2}^K\exp\left(-\frac{c^2\alpha_{i}^2t^2}{16\eta_{t+1}^2}\right)\\ 
	&=\sum_{i=2}^K \exp\left(-\frac{c^2\alpha_{i}^2 t}{16\alpha^2}\right). \numberthis \label{eq:prob_ftpl_comp}
	\end{align*}
	Therefore, the expected cost under FTPL can be bounded as 
	\begin{align*}
	\E[\mathcal{C}^{\text{FTPL}}&(T,r)]	\\
	\le& \sum_{t=1}^{T}(c+M)p_{h,t}+r_t\\
	=& R_T+\sum_{t=1}^{T}(c+M)p_{h,t}\\
	\le& R_T+ (c+M) T'+(c+M)\sum_{t=\lceil T'\rceil}^T p_{h,t}  \\
	\overset{(a)}{\le}& R_T+ (c+M) T'\\
	&+(c+M)\sum_{t=\lceil T'\rceil}^T \sum_{i=2}^K \exp\left(-\frac{c^2\alpha_{i}^2t}{16\alpha^2}\right)  \\
	\le& \left((3+2M/c)\max_{\alpha_{i}\ne 0}\frac{1-g(\alpha_{i})}{\alpha_{i}}\right) R_T\\ &+(c+M)\sum_{i=2}^K\frac{16\alpha^2}{c^2\alpha_{i}^2}.
	\end{align*}	
	Here $(a)$ follows from \eqref{eq:prob_ftpl_comp}. If $T<T'$, then the bound will only have the first term and even in that case the above bound holds. Therefore the above bound holds for  any $T$.
\end{proof}

\begin{proof}[Proof of Theorem \ref{thm:comp_adv}(a)]
		By using Lemma \ref{lem:opt_offline_cost_LB}, we have
	\begin{align*}
		\sigma_{A}^{\text{FTPL}}\le& \sup_{r} \frac{\E[\mathcal{C}^{\text{FTPL}}(T,r)]}{(\min_i{c\alpha_{i}+g(\alpha_{i})\kappa})R_T/\kappa^2}\\
		\overset{(a)}{\le}& \frac{\kappa^2(3+2M/c)}{\min_i(c\alpha_{i}+g(\alpha_{i}) \kappa)}\max_{\alpha_{i}\ne 0}\frac{1-g(\alpha_{i})}{\alpha_{i}} \\
		&+\sup_r \frac{\kappa^2(M+c)}{R_T\min_i(c\alpha_{i}+g(\alpha_{i}) \kappa)}\sum_{i=2}^K\frac{16\alpha^2}{c^2\alpha_{i}^2}\\
		\le &\frac{\kappa^2(3+2M/c)}{\min_i(c\alpha_{i}+g(\alpha_{i}) \kappa)}\max_{\alpha_{i}\ne 0}\frac{1-g(\alpha_{i})}{\alpha_{i}} \\
		&+ \frac{\kappa^2(M+c)}{\min_i(c\alpha_{i}+g(\alpha_{i}) \kappa)}\sum_{i=2}^K\frac{16\alpha^2}{c^2\alpha_{i}^2}.
	\end{align*}
	Here $(a)$ is obtained by using Lemma \ref{lem:FTPL_comp_adv}.
\end{proof}

\subsection{Proof of Theorem \ref{thm:comp_adv}(b)}
\begin{proof}[Proof of Theorem \ref{thm:comp_adv}(b)]
	By using Lemma \ref{lem:opt_offline_cost_LB}, we have
	\begin{align*}
	\sigma^{\text{W-FTPL}}_A \le& \sup_{r} \frac{R_{T_s}+\E[\mathcal{C}_{A}^{\text{FTPL}}(T,r)]}{R_T(\min_i\{c\alpha_i+g(\alpha_{i})\kappa\})/\kappa^2}\\
	\overset{(a)}{\le} & \frac{\kappa^2(3+2M/c)}{\min_i(c\alpha_{i}+g(\alpha_{i}) \kappa)}\max_{\alpha_{i}\ne 0}\frac{1-g(\alpha_{i})}{\alpha_{i}}\\
	&+\frac{\kappa^2(M+c)}{\min_i(c\alpha_{i}+g(\alpha_{i}) \kappa)}\sum_{i=2}^K\frac{16\alpha^2}{c^2\alpha_{i}^2}\\
	&+\frac{\kappa^2}{\min_i(c\alpha_{i}+g(\alpha_{i}) \kappa)}.
	\end{align*}
	Here $(a)$ is obtained by using Theorem \ref{thm:comp_adv}(a) and the fact that $R_{T_s}\le R_T$.
\end{proof}

\subsection{Proof of Theorem \ref{thm:comp_stch}(a)} 
\begin{proof}[Proof of Theorem \ref{thm:comp_stch}(a)]

	\begin{align*}
		\sigma^{\text{RR}}_{S}(T)&=\frac{\E[\mathcal{C}^{\text{RR}}(T,r)]}{\min_i\{c\alpha_iT+g(\alpha_i)\mu T+M\alpha_i\}}\\  
		&= \frac{\mathcal{R}^{\text{RR}}_S(T)+\min_i\{c\alpha_iT+g(\alpha_i)\mu T+M\alpha_i\}}{\min_i\{c\alpha_iT+g(\alpha_i)\mu T+M\alpha_i\}}\\
		& \le 1+\frac{\mathcal{R^{\text{RR}}}_S(T)}{\mu_{i^\star}T}.
	\end{align*}

	By using Theorem \ref{thm:reg_stch} we have,
	\begin{align*}
	\sigma^{\text{RR}}_{S}(T)\ge 1+  \frac{M}{\mu_{i^\star}}\left(\frac{1}{\frac{M}{\kappa-c\alpha_{i'}-g(\alpha_{i'})\kappa}+ \frac{M}{c} +2}-\frac{1}{T}\right) pd.
	\end{align*}
	For large $T$ we get,
	\begin{align*}
	\sigma^{\text{RR}}_{S}(T)> 1.
	\end{align*}
\end{proof}

\subsection{Proof of Theorem \ref{thm:comp_stch}(b)} 

\begin{proof}[Proof of Theorem \ref{thm:comp_stch}(b)]
	By using Theorem \ref{thm:reg_stch} we have
	\begin{align*}
	\sigma^{\text{FTPL}}_{S}(T) &\le  1+\frac{1}{\mu_{i^\star}T}(16\alpha^2+1)\left( \frac{M}{\Delta_{\text{min}}^2}+\sum_{i\ne i^\star}\frac{1}{\Delta_i}\right)\\
	&=1+\OO(1/T).
	\end{align*}
\end{proof}

\subsection{Proof of Theorem \ref{thm:comp_stch}(c)} 
\begin{proof}[Proof of Theorem \ref{thm:comp_stch}(c)]
	By  using Theorem \ref{thm:reg_stch} we have $\mathcal{R}^{\text{W-FTPL}}_S(T)$ as  a constant and 
	\begin{align*}
	\sigma^{\text{W-FTPL}}_{S}(T) \le  1+\frac{\mathcal{R}^{\text{W-FTPL}}_S(T)}{T\mu_{i^\star}}.
	\end{align*}
	Therefore we get,
	$$\sigma^{\text{W-FTPL}}_S(T) \le 1+\OO(1/T). $$
	
\end{proof}

\bibliography{ref_PC}
\bibliographystyle{elsarticle-num}
\section{Appendix}
\label{sec:appendix}
\begin{proof}[Proof of Lemma \ref{lem:adv_stch_reg}] For i.i.d. stochastic arrivals, we have
\begin{align*}
\mathcal{R}^{\mathcal{P}}_S&(T)\\
&=\E_{\mathcal{P},r}[\mathcal{C}^\mathcal{P}(T,r)] -\min_i\{c\alpha_iT+g(\alpha_i)\mu T+M\alpha_i\}\\
&=\E_{\mathcal{P},r}[\mathcal{C}^\mathcal{P}(T,r)] -\min_i\{\E_r[c\alpha_iT+g(\alpha_i)R_T+M\alpha_i]\}\\
&\overset{(a)}{\le} \E_{\mathcal{P},r}[\mathcal{C}^\mathcal{P}(T,r)] -\E_r[\min_i\{c\alpha_iT+g(\alpha_i)R_T +M\alpha_i\}]\\
&=\E_r[\mathcal{R}_A(T,r)] \le \sup_{r} \left(\mathcal{R}^{\mathcal{P}}_A(T,r) \right) = \mathcal{R}_A^\mathcal{P}(T),
\end{align*}
where $(a)$ is obtained by Jensen's inequality.
\end{proof}

\begin{lemma}
	\label{lem:min_E}
	If $X$ is a random variable and $\E[f_1(X)]=\E[f_2(X)]=m$, then 
	$$m-\E[\min\{f_1(X),f_2(X)\}]=\frac{1}{2}\E[|f_1(X)-f_2(X)|].$$
\end{lemma}
\begin{proof}
	\begin{align*}
		\min\{X,Y\}=\frac{1}{2}(X+Y)-\frac{1}{2}(|X-Y|)\\
	\implies m-\min\{f_1(X),f_2(X)\}=m-\frac{1}{2}(f_1(X)+f_2(X))\\
	+\frac{1}{2}(|f_1(X)-f_2(X)|)\\
		\implies m-\E[\min\{f_1(X),f_2(X)\}]=\frac{1}{2}\E[|f_1(X)-f_2(X)|].
	\end{align*}
\end{proof}

\begin{lemma}
	\label{lem:binomial_t}
	If $X\sim \text{Bin}(T,p)$ then 
	\begin{align*}
	\frac{1}{2}\E[|X-Tp|]\ge &\frac{\sqrt{Tp(1-p)}}{e\sqrt{2\pi}} \left[1+\frac{1}{12T+1}-\frac{1}{12\lfloor Tp\rfloor} \right. \\
	&\left.- \frac{1}{12(T-\lfloor Tp\rfloor-1)} \right].
	\end{align*}
\end{lemma}
\begin{proof}
	By using equation (1) from \cite{berend2013sharp}  we get
	\begin{align*}
		\E[|X&-Tp|]\\
		=&2(1-p)^{T-\lfloor Tp\rfloor} p^{\lfloor Tp\rfloor+1}(\lfloor Tp\rfloor+1) {T\choose \lfloor Tp+1\rfloor}\\
		=&2(1-p)^{T-\lfloor Tp\rfloor} p^{\lfloor Tp\rfloor+1} \frac{T!}{\lfloor Tp\rfloor! (T-\lfloor Tp\rfloor-1)! }\\
		\overset{(a)}{\ge}&2(1-p)^{T-\lfloor Tp\rfloor} p^{\lfloor Tp\rfloor+1} \frac{1}{e\sqrt{2\pi}}\frac{T^{T+\frac{1}{2}}e^{\frac{1}{12T+1}}} {\lfloor Tp\rfloor^{\lfloor Tp\rfloor+\frac{1}{2}}e^{\frac{1}{12\lfloor Tp\rfloor}}}\\
		&\times \frac{1}{(T-\lfloor Tp\rfloor-1)^{T-\lfloor Tp\rfloor-1+\frac{1}{2}}e^{\frac{1}{12(T-\lfloor Tp\rfloor-1)}}}\\
		\overset{(b)}{\ge}& \frac{2}{e\sqrt{2\pi}}(1-p)^{T-\lfloor Tp\rfloor} p^{\lfloor Tp\rfloor+1} \frac{T^{T+\frac{1}{2}}e^{\frac{1}{12T+1}}} {(Tp)^{\lfloor Tp\rfloor+\frac{1}{2}}e^{\frac{1}{12\lfloor Tp\rfloor}}}\\
		&\times \frac{1}{(T-Tp)^{T-\lfloor Tp\rfloor-\frac{1}{2}}e^{\frac{1}{12(T-\lfloor Tp\rfloor-1)}}}\\
		= & \frac{2\sqrt{Tp(1-p)}}{e\sqrt{2\pi}} e^{\frac{1}{12T+1}- \frac{1}{12\lfloor Tp\rfloor}- \frac{1}{12(T-\lfloor Tp\rfloor-1)}}\\
		\overset{(c)}{\ge} &\frac{2\sqrt{Tp(1-p)}}{e\sqrt{2\pi}} \left[1+\frac{1}{12T+1}-\frac{1}{12\lfloor Tp\rfloor} \right. \\
		&\left.- \frac{1}{12(T-\lfloor Tp\rfloor-1)}\right],
	\end{align*}
	where $(a)$ follows by using the result in \cite{robbins1955remark} which is $\sqrt{2\pi}n^{n+\frac{1}{2}}e^{-n}e^{\frac{1}{12n+1}} < n! <\sqrt{2\pi}n^{n+\frac{1}{2}}e^{-n}e^{\frac{1}{12n}}$. $(b)$ follows by using the fact that $x-1\le \lfloor x \rfloor\le x$. $(c)$ follows from the fact $e^x\ge 1+x$.
\end{proof}
\end{document}